\documentclass[]{interact}

\usepackage{amsmath, array, dsfont, epsfig, flexisym, tabulary, bm, subfig}
\usepackage{amsthm}
\usepackage{multirow}
\usepackage{rotating}
\usepackage{lscape}
\usepackage{booktabs}
\usepackage{epstopdf}
\usepackage{hyperref}
\usepackage{color}

\usepackage[super,sort&compress]{natbib}
\setcitestyle{comma,numbers,super,open={},close={}} 
\makeatletter
\renewcommand\@biblabel[1]{#1.}
\makeatother

\theoremstyle{plain}
\newtheorem{theorem}{Theorem}[section]
\newtheorem{lemma}[theorem]{Lemma}

\theoremstyle{definition}

\theoremstyle{remark}

\newcommand{\T}{^{\prime}}

\begin{document}


\title{Robust Estimation for Linear Panel Data Models}

\author{
\name{Beste Hamiye Beyaztas\textsuperscript{a} and Soutir Bandyopadhyay\textsuperscript{b} \thanks{CONTACT Beste Hamiye Beyaztas. Email: beste.sertdemir@medeniyet.edu.tr}}
\affil{\textsuperscript{a} Department of Statistics, Istanbul Medeniyet University, Goztepe-Northern Campus, 34720, Istanbul, Turkey. \\
\textsuperscript{b} Department of Applied Mathematics and Statistics, Colorado School of Mines, Golden, CO 80401, USA.}
}

\maketitle

\begin{abstract}
In different fields of applications including, but not limited to, behavioral, environmental, medical sciences and econometrics, the use of panel data regression models has become increasingly popular as a general framework for making meaningful statistical inferences. However, when the ordinary least squares (OLS) method is used to estimate the model parameters, presence of outliers may significantly alter the adequacy of such models by producing biased and inefficient estimates. In this work we propose a new, weighted likelihood based robust estimation procedure for linear panel data models with fixed and random effects. The finite sample performances of the proposed estimators have been illustrated through an extensive simulation study as well as with an application to blood pressure data set. Our thorough study demonstrates that the proposed estimators show significantly better performances over the traditional methods in the presence of outliers and produce competitive results to the OLS based estimates when no outliers are present in the data set.
\end{abstract}


\section{Introduction} \label{Sec:1}
Panel data, also known as longitudinal data in biological sciences, are two-dimensional data in which cross-sectional measurements are observed over time. 
These type of data typically allow us to take into account the unobserved individual-specific heterogeneity as well as the intra-individual dynamics (cf. \cite{Baltagi2005} and \cite{Hsiao1985} for more details), and therefore, in general, are more informative and yield more degrees of freedom, less collinearity between the variables and more efficiency than a single cross-sectional or time-series data, thereby improving the accuracy and precision in the inference of model parameters.  

Since the seminal paper of \cite{Balestra1966}, panel data have received growing attention in many empirical and methodological studies. As pointed out in \cite{Hsiao2007}, the main sources leading to the improvements in panel data studies 
include (i) increased availability of such data, (ii) better capability to model the complexity of human behavior than a pure cross-section or time series data, and (iii) demanding methodology. In this context, the linear panel data regression models have become most widely applied statistical methods to analyze two-dimensional data in many fields, such as econometrics, biostatistics etc. For a comprehensive review on static linear panel data models and its applications in different areas see \cite{Baltagi2005}, \cite{Fitzmaurice2004}, \cite{Greene2003}, 
  \cite{Gardiner2009}, \cite{Laird1982}, \cite{Maddala1973}, \cite{Mundlak1978}, \cite{Diggle2002}, \cite{Hill2007}, \cite{Wallace1969}, \cite{Wooldridge2002}, and the references therein. 

In panel data studies, three main sources of variability are generally considered, namely, (i) the within variation, i.e., the variation from observation to observation in each of cross-sectional unit,  (ii) the between variation, i.e., the variation in observations from an individual unit to another individual unit, and (iii) the overall variation, i.e., the variation over both dimensions, since panel data include the information over two dimensions, cross-sectional and time series (cf. \cite{Bala2014} and \cite{Kennedy2003} for details). 
The statistical appeal of panel data models 
 typically lies in the fact that these models focus particularly on explaining within variations over time and provide controls over individual heterogeneity. The most commonly used panel data models are fixed and random effects models 
(cf. \cite{Bala2014} and \cite{Zhang2010}). In the fixed effect approach, subject-specific means (individual-specific effects, individual heterogeneity), which belong to each cross-sectional unit, are assumed to be fixed and are included as time-invariant intercept terms in the regression model, while these may vary across subjects. On the other hand, in random effect models, individual heterogeneity is explained by the differences in the error variance components. As noted in \cite{Mundlak1978}, the main difference between the fixed and random effects models is that the fixed effects model assumes that the time-invariant characteristics of individuals are correlated with the covariates, whereas random effects model does not allow such correlation. 

Typically, one uses OLS methods for making statistical inferences regarding the parameters of linear panel data regression models. However, to obtain consistent estimates of the model parameters, traditional estimation techniques require some assumptions such as strict exogeneity with respect to the error terms and homoscedasticity of the error terms, which are rarely fulfilled in practice. Hence, the classical OLS estimators may considerably be affected due to any departure from the model assumptions as well as the presence of outliers. The outlying observations are generally masked due to the complex nature of the data and not directly detectable using standard outlier diagnostics. Moreover, the OLS based estimators are highly sensitive to the leverage points due to the distortions being caused by the outliers in the covariates. Thus, the well-known estimators, such as generalized least squares (GLS) estimator for random effects model and fixed effects estimators based on several transformations, may lead us to incorrect and unreliable results. To overcome these issues, \cite{Bramati2007} have considered alternatives to the fixed effect estimator for the purpose of building highly robust procedures with high breakdown point. More recently, \cite{Aquaro2013} have proposed a new estimation procedure based on two different data transformations by applying standard robust estimation methods in the fixed effects linear panel data framework. A robust algorithm based on the idea of weighting down the large order statistics of squared residuals has been proposed in \cite{Visek2015} to obtain reliable estimates of the model parameters. To the best of our knowledge, only a few studies considering the robustness of conventional estimation methods are available in the context of static linear panel data models; see, for instance, \cite{Aquaro2013}, \cite{Bramati2007}, \cite{Namur2011}, \cite{Visek2015} and \cite{Wagenvoort2002}.

This paper aims to study the impacts of outlying observations on the OLS based estimation methods (such as between, pooled OLS, fixed effects and random effects estimators as discussed in Section~\ref{Sec:2}) in linear panel data models and suggest robust alternatives to these estimation procedures. 
The proposed weighted likelihood based estimators, based on weighted likelihood estimating equations introduced in \cite{Markatou1997}, produce more robust estimates compared to their traditional counterparts in the presence of outlier(s) or in case of any departure from model assumptions and their asymptotic properties are equivalent to the OLS based techniques when no outliers are present in the data. In this study, we focus on the impacts of several types of outliers including vertical outliers and leverage points (random and concentrated) on the estimation procedures. Monte Carlo experiments under different data generating processes and contamination schemes are used to compare the finite sample performances of the proposed estimators and traditional OLS based estimators. The numerical results support that the proposed methods yield more accurate and precise estimates compared to the OLS estimators when the data have outliers. 

The rest of the paper is organized as follows. We start with providing details about the static linear panel data models and discuss the OLS based estimation methods commonly used to estimate the parameters (cf. Section~\ref{Sec:2.1}). In Section~\ref{Sec:2.2} we describe the estimation method based on weighted likelihood and subsequently propose the robust counterparts of the OLS estimators. The finite sample properties of the proposed methods are illustrated through an extensive simulation study and the results are compared with traditional estimation methods in Section~\ref{Sec:3}. To further validate the applicability of our proposed methods, we apply those to blood pressure data. The results are presented in Section \ref{Sec:4}. 

\section{Linear Panel Data Models} \label{Sec:2}
Let us consider the linear panel data regression model with a random sample $\left\lbrace \left( y_{it}, x_{it}, \alpha_i \right), i = 1, \ldots, N; t = 1, \ldots, T \right\rbrace$ as follows.
%
\begin{equation*}
y_{it} = x_{it}^{\prime} \beta + \alpha_i + \varepsilon_{it},  
\end{equation*}
%
where the subscript $i$ represents an individual observed at time $t$, $\alpha_i$'s are the unobserved individual-specific effects (time-invariant characteristics), $\beta$ is a $K \times 1$ vector of coefficients and an element of the parameter space $\Theta$, $y_{it}$ and $x_{it}$'s are the response variable and the $K$-dimensional vector of explanatory variables, respectively and $\varepsilon_{it}$'s are the independent and identically distributed (iid) error terms with $E \left( \varepsilon_{it} \vert x_{i1}, \ldots, x_{iT}, \alpha_i \right) = 0$, $E \left( \varepsilon_{it}^2 \vert x_{i1}, \ldots, x_{iT}, \alpha_i \right) = \sigma_{\varepsilon}^2$ and $E \left( \varepsilon_{it} \varepsilon_{is} \vert x_{i1}, \ldots, x_{iT}, \alpha_i \right) = 0$ for $t \neq s$. The above panel data regression model can be represented in matrix form as follows.
\begin{equation*} 
y = \alpha \otimes e_T + X \beta + \varepsilon,
\end{equation*}
where $y = \left( y_{1}, \ldots, y_{N} \right)^{\prime}$ is an $NT \times 1$ vector obtained by stacking observations $y_i = \left( y_{i1}, \ldots, y_{iT} \right)\T$ for individual $i=1,\ldots,N$, $X = \left( x_1, \ldots, x_N \right)^{\prime}$ is an $NT \times K$ matrix of regressors with $x_i = \left( x_{i1}^{\prime}, \ldots, x_{iT}^{\prime} \right)^{\prime}$, $\alpha$ is an $N \times 1$ vector consisting of the individual effects $\alpha_i$ for $i = 1 \ldots N$, $e_T$ is a $T \times 1$ vector of ones and $\otimes$ denotes the kronecker product. 

In fixed effects models, only variation within each cross-sectional unit is exploited (cf.~ \cite{Bala2014}, \cite{Kennedy2003} and \cite{Zhang2010}). Thus, in the presence of small or no within variation, the coefficients of the regressors in fixed effects models cannot be correctly estimated or identified, as noted in \cite{Cameron2009}.  
The fixed effects models generally allow for possible correlations between individual-specific unobservable effects and independent variables by including dummy variables for different intercepts, allowing a limited form of endogeneity (cf. \cite{Cameron2009}) while yielding unbiased estimates of the regression parameters (cf. \cite{Greene2003}, \cite{Kennedy2003}, \cite{Zhang2010}, and \cite{Lindeboom2002}).  
The information on both within and between variations are included by the random effects models. In random effects models, the individual-specific effects are being included in the model as a part of the disturbance, and these are required to be uncorrelated with the explanatory variables and the error terms (cf. \cite{Greene2003} and \cite{Maddala1973}). In particular, if $\alpha_i$ is assumed to be random then, the random effects model can be formulated as follows.  %

\begin{equation} \label{Eq:3}
y_{it} = x_{it}^{\prime} \beta + \alpha_i + \varepsilon_{it} = x_{it}^{\prime} \beta + \nu_{it},~~\alpha_i \sim \mathrm{iid}(0, \sigma_{\alpha}^2),~~\varepsilon_{it} \sim \mathrm{iid}(0, \sigma_{\varepsilon}^2),
\end{equation}
where $\nu_{it} = \alpha_i + \varepsilon_{it}$ denotes a compound error term with $\sigma_{\nu}^2 = \sigma_{\alpha}^2 + \sigma_{\varepsilon}^2$ and $cov\left( \nu_{it},\nu_{is}\right) = \sigma_{\alpha}^2$ for $t \neq s$. $\alpha_i$'s are assumed to be uncorrelated with $\varepsilon_{it}$ and $x_{it}$.

The pooled regression model
\begin{equation*} 
y_{it} = \alpha + x_{it}^{\prime} \beta  + \varepsilon_{it},  
\end{equation*}
 is a restricted type of panel data model such that the regression coefficients, i.e., $\alpha$ and $\beta$, have the common values to all cross-sectional units for all time periods as noted in \cite{Zhang2010} and \cite{Hsiao2003}. 

Finally, before we describe the between regression models, let $\bar{y}_i = T^{-1} \sum_{t = 1}^T y_{it}$, $\bar{x}_i = T^{-1} \sum_{t = 1}^T x_{it}$ and $\bar{\varepsilon}_i = T^{-1} \sum_{t = 1}^T \varepsilon_{it}$, respectively, denote the time averages of $y_{it}$, $x_{it}$ and $\varepsilon_{it}$ for the $i$-th cross-sectional unit. By considering the $N$ linear regression models based on the time averages of each cross-sectional unit, the between model is defined as follows:
\begin{equation} \label{Eq:10}
\bar{y}_i = \alpha_{i} + \bar{x}_i \beta + \bar{\varepsilon}_i.
\end{equation}
The between regressions are frequently used to investigate the long-run relationships by ignoring all the information owing to intra-subject variability (cf. \cite{Cameron2009}, \cite{Baltagi1984}  and \cite{Croissant2008}). For example, \cite{Houthakker1965} has examined the elasticity of demand for some countries and compared the estimates obtained using within and between regressions. The results obtained from the between country model can be interpreted as long run effects whereas the short run effects are captured by the within country regression model. Additionally, \cite{Baltagi1984} have reported the results of elasticity estimates for short run price and long run price, and compared the estimates in terms of mean, standard deviation and root mean square error (RMSE) criteria. It has been emphasized that the between estimator has a better performance according to the RMSE criterion for estimates of long run elasticity price than that of short run price elasticity.  

Next, we briefly discuss the commonly used estimation procedures for above mentioned linear panel data models.
\subsection{Traditional Estimation Methods} \label{Sec:2.1}

The estimation procedures commonly applied in linear panel data models discussed above, can be examined within the scope of OLS estimation as noted in \cite{Croissant2008}. The OLS based estimation techniques mainly rely on the type of variations (cf. \cite{Kennedy2003}). 

The pooled OLS estimator is simply the implementation of the OLS method to the linear model on the pooled data across two dimensions by completely disregarding the panel structure of the data. Therefore, $\widehat{\beta}_{pols}$, the pooled OLS estimator of $\beta$, can be obtained as follows.
\begin{equation*} 
\widehat{\beta}_{ pols} = \left( \sum_{i=1}^N \sum_{t=1}^T x_{it}^{\prime} x_{it} \right)^{-1} \left( \sum_{i=1}^N \sum_{t=1}^T x_{it}^{\prime} y_{it} \right).
\end{equation*}

As noted in \cite{Greene2003}, \cite{Wooldridge2002} and \cite{Cameron2009}, the pooled OLS method provides consistent estimates of the parameters for random effects and pooled regression models under the independence assumption of explanatory variables and error terms. On the other hand, it is severely biased and inconsistent for the fixed effects model due to the inclusion of the individual-specific effects, which are correlated with the explanatory variables. Also,  while investigating the bias and efficiency of some well-known panel data estimators in observational health studies, \cite{Dieleman2014} raised two main concerns for pooled estimator, namely, the heteroscedastic error terms and the bias caused by the omitted individual-specific effects. 

The estimation of the fixed effects model requires the time-demeaned  data. The fixed effects transformed model for the mean-centered data is obtained as follows.
\begin{equation*} 
\ddot{y}_{it} = \ddot{x}_{it}^{\prime} \beta + \ddot{\varepsilon}_{it},
\end{equation*}
where $\ddot{y}_{it} = y_{it} - \bar{y}_i$, $\ddot{x}_{it} = x_{it} - \bar{x}_i$ and $\ddot{\varepsilon}_{it} = \varepsilon_{it} - \bar{\varepsilon}_i$, so that the individual-specific effects have been eliminated. 
Under the assumptions of fixed effects model, $\widehat{\beta}_{fe}$, the fixed effects estimator of $\beta$, can be obtained as
\begin{equation*} 
\widehat{\beta}_{fe} = \left( \sum_{i=1}^N \sum_{t=1}^T \ddot{x}_{it}^{\prime} \ddot{x}_{it} \right)^{-1} \left( \sum_{i=1}^N \sum_{t=1}^T \ddot{x}_{it}^{\prime} \ddot{y}_{it} \right).
\end{equation*}
The fixed effects estimator (also known as the within estimator) is consistent for the fixed effects model when time dimension, $T$ gets large (cf. \cite{Baltagi2005}). Further, as noted in \cite{Allison2009}, the precisions of the fixed effects estimates are significantly affected when the independent variables vary greatly across individual units and simultaneously exhibit small variation over time for each individual. 

The fixed effects least squares method has several shortcomings. Firstly, the fixed effects estimators suffer from the incidental parameter problem (see, \citep{Neyman1948} for more details). If the cross-sectional dimension, $N$, is significantly large, the fixed effects estimators of individual-specific effects become biased and inconsistent because of the increasing number of these parameters (cf. \cite{Baltagi2005}, \cite{Greene2003} and \cite{Lancaster2000}). Furthermore, the fixed effects method is incapable of estimating the coefficients of the time-invariant variables since it only considers the variation within cross-sections. Therefore, the explanatory power of the model decreases with less efficient estimates (see, \cite{Allison2009} and \cite{Cameron2009} for more details). Another possible drawback is that several dummies used for time-invariant variables such as gender, race, geographic location, education, religion cause to aggravate collinearity among the regressors as noted in \cite{Baltagi2005}. Furthermore, \cite{Griliches1986} emphasize that the fixed effects methods underestimate the model parameters and result in drastically biased inference since the measurement errors get magnified in within dimension. The random effects model compensates for some of these problems encountered in fixed effects least squares.

The GLS method is used for estimating random effect model to deal with the autocorrelation in the error terms caused by the individual-specific effects. This method entails the quasi demeaning transformation of the variables to obtain the homoscedastic variance-covariance matrix for achieving efficiency, as noted in \cite{Croissant2008} and \cite{Jirata2014}. As emphasized in \cite{Baltagi2005}, the GLS method  asymptotically provides the best linear unbiased estimator if the variance-covariance matrix of the disturbance term is known. Also, \cite{Croissant2008} points out that it produces the equivalent estimates of $\beta$ to the OLS method on the quasi-demeaned data. The quasi-demeaning transformation contains subtracting the time-averages, weighted by using the variances of the idiosyncratic errors and individual effects, from the original variables. The transformed version of the random effects model is expressed as,
\begin{equation*} 
\tilde{y}_{it} = \tilde{x}_{it}^{\prime} \beta + \tilde{\nu}_{it},
\end{equation*}
where $\theta = 1 - \left[ \frac{\sigma_{\varepsilon}^2}{\sigma_{\varepsilon}^2 + T \sigma_{\alpha}^2} \right]^{1/2}$, $\tilde{\nu}_{it} = \nu_{it} - \theta \bar{\nu}_i$, $\tilde{y}_{it} = y_{it} - \theta \bar{y}_i$ and  $\tilde{x}_{it} = x_{it} - \theta \bar{x}_i$ with the time averages $\bar{y}_i$ and $\bar{x}_i$. By running OLS method on the transformed model, the GLS estimator (also called as random effects estimator), $\widehat{\beta}_{re}$ can be obtained as follows.
\begin{equation*} 
\widehat{\beta}_{re} = \left( \sum_{i=1}^N \sum_{t=1}^T \tilde{x}_{it}^{\prime} \tilde{x}_{it} \right)^{-1} \left( \sum_{i=1}^N \sum_{t=1}^T \tilde{x}_{it}^{\prime} \tilde{y}_{it} \right).
\end{equation*}
Note that, the fixed effects and pooled OLS estimators can be obtained by employing OLS method on the transformed model for $\theta = 1$ with $\sigma_{\varepsilon}^2 = 0$ and $\theta = 0$ with $\sigma_{\alpha}^2 = 0$, respectively, as the special cases of the above mentioned GLS estimator. 

When the random effects model as in Eq. \ref{Eq:3} is appropriate, both fixed effects and random effects estimators are consistent but random effects method provides more efficient estimates, with high explanatory power, compared to the fixed effects method. This is due to fact that the random effects estimators have the advantages of using both within and between variations, hence, these can be viewed as a weighted average of the between and fixed effects estimators (cf. \cite{Cameron2009} and \cite{Kennedy2003}). However, there is a trade-off between bias and efficiency, and the random effects method is more vulnerable to omitted variable bias than the fixed effects method, as noted in \cite{Allison2009}. In case of no omitted variables, the random effects model is generally preferred over the fixed effects model because it allows for estimating the effects of time-invariant variables; see \cite{Kennedy2003}.

The between regression models (Eq. \ref{Eq:10}) include the information reflected in the differences between cross-sections. A large between variation generally indicates the differences in means of the variables  over time for each subject as noted in \cite{Houthakker1965}. By employing the OLS method on the between regression model, the between estimator, $\widehat{\beta}_{be}$ is obtained as follows. 
\begin{equation*} 
\widehat{\beta}_{be} = \left( \bar{x}_i^{\prime} \bar{x}_i \right)^{-1} \left( \bar{x}_i^{\prime} \bar{y}_i \right).
\end{equation*}
Although the between estimator generates consistent results for the pooled and random effects models, it is rarely preferred in practice since the pooled and random effects estimators yield more efficient results compared to the between estimator, see \cite{Cameron2009}. Further, as noted in \cite{Zhang2010}, it can be used to estimate the effects of time-invariant variables in fixed effects model, but with biased estimates of the effects of both time-invariant and time-variant variables.

In spite of the fact that all of the traditional methods discussed above suffer heavily due to the presence of outlying observations, the existing literature on the robust methods to estimate static panel data models is fairly limited. Recently, a few different approaches within the robust estimation framework for the fixed effects panel data models have been developed by utilizing the generalized M-estimation  and least trimmed squares (LTS) techniques; see, for example, \cite{Bramati2007}, and \cite{Aquaro2013}. \cite{Bramati2007} defined the robust versions of fixed effects estimator with high breakdown point by extending some known robust regression estimators, such as LTS estimator of \cite{Rousseeuw1984} and a combination of M and S estimates of \cite{Maronna2000}. Another robust estimation approach has been proposed in \cite{Aquaro2013} based on two different data transformations by employing the efficient weighted least squares estimator of \cite{Gervini2002} and the reweighted LTS estimator of \cite{Cizek2010} in the context of linear regression model.

Next, in Section \ref{Sec:2.2}, we propose robust alternatives of the OLS based estimation procedures, which are highly sensitive to the presence of outliers, erroneous observations and any departure from the distributional assumptions on the error terms. Our approach is primarily based on the weighted likelihood estimating equations methodology introduced in \cite{Markatou1997}. The main idea behind the weighted likelihood methodology is to replace the maximum likelihood (ML) equations with weighted score equations, in which the weights come from minimum disparity estimation as in \cite{Lindsay1994}, for obtaining efficient estimates and reducing the effects of outliers on the score equations. Note that, the ML method (and its weighted version) provides a flexible framework for the purposes of likelihood based model specification testing and estimation in the presence of endogeneity problem leading to correlation between regressors and error terms. Also, it eliminates the incidental parameters problem over time, see \cite{Bai2014}. Furthermore, the ML estimator is equivalent to GLS estimator under the assumptions of homoscedasticity, no-autocorrelation and normally distributed error terms (cf.~\cite{Aitken1935}). 

\subsection{New Robust Weighted Likelihood based Estimation Procedure} \label{Sec:2.2}

As defined earlier, let $\lbrace y_1, \ldots, y_N \rbrace^{\prime}$ be an iid sample of $NT \times 1$ vector, and $\lbrace x_1, \ldots, x_N \rbrace^{\prime}$ denote an $NT \times K$ matrix of predictors with $x_i = \lbrace x_{i1}^{\prime}, \ldots, x_{iT}^{\prime} \rbrace^{\prime}$. Let us consider the random effects model given in Eq.~(\ref{Eq:3}) with density function $f = f( y_{it}; x_{it}, \beta)$ and define the joint probability density for disturbance terms, $\varepsilon_i + \alpha_i e_T  = y_i - x_i \beta $, as given below.
\begin{eqnarray*} 
f \left( \varepsilon_i + \alpha_i e_T \right) = \left( 2\pi \right)^{-\frac{T}{2}} \vert \mathbf{\Omega} \vert ^{-\frac{1}{2}} \exp{ \left\lbrace -\frac{1}{2} \left( y_i - x_i \beta \right)^{\prime} \mathbf{\Omega}^{-1} \left( y_i - x_i \beta \right) \right\rbrace}.
\end{eqnarray*}
Then, under the assumption of normally distributed $\nu_{it}$ and $\alpha_i$ terms, the log likelihood function utilizing the log likelihood contribution for cross-sectional unit $i$, $\mathcal{L}_i \left( \beta \right) = \log{f \left( y_{i}; x_{i}, \beta \right)}$,  can be expressed as
\begin{eqnarray*} 
\log{\mathcal{L} \left( \beta, \sigma_{\varepsilon}^2, \sigma_{\alpha}^2 \right)} &=& \sum_{i=1}^N \mathcal{L}_i \left( \beta \right) = \sum_{i=1}^N \left[ \sum_{t=1}^T \log{f \left( y_{it}; x_{it}, \beta \right)} \right],\nonumber  \\
&=& -\frac{NT}{2} \log{(2\pi) } -\frac{N}{2} \log{\vert \mathbf{\Omega} \vert} -\frac{1}{2} \sum_{i=1}^N \left( y_i - x_i \beta \right)^{\prime} \mathbf{\Omega}^{-1}  \left( y_i - x_i \beta \right),\nonumber\\ 
\end{eqnarray*}
where $ \mathbf{\Omega} = E \left( \nu_i \nu_i^{\prime} \right) = \sigma_{\varepsilon}^2 \mathbf{I}_T + \sigma_{\alpha}^2 e_T e_T^{\prime} $ denote a $T \times T$ matrix, with $\mathbf{I}_T$ being an identity matrix of dimension $T$. Let $V$ denote the full $NT \times NT$ variance-covariance matrix of compound error terms $\nu_{it}$, i.e., 
 \[ V =
  \begin{pmatrix}
    \mathbf{\Omega} & 0 & \dots & 0 \\
    0 & \mathbf{\Omega} & \dots & 0 \\
    \vdots & \vdots & \ddots & \vdots \\
    0 & 0 & \dots & \mathbf{\Omega}
  \end{pmatrix}
 = \mathbf{I}_N \otimes \mathbf{\Omega}. \]
The ML estimator of the unknown parameter vector, $\widehat{\beta}$ is obtained by solving the score functions
\begin{equation*} 
\underset{\beta \in \mathbb{R}^{K}}{\mathrm{arg~max}}\prod_{i=1}^N f \left( y_{i}; x_{i}, \beta \right) =  \underset{\beta \in \mathbb{R}^{K}}{\mathrm{arg~min}} \sum_{i=1}^N r_i^2 \left( \beta \right) = \underset{\beta \in \mathbb{R}^{K}}{\mathrm{arg~min}} \sum_{i=1}^N \left( y_i - x_i \beta \right)^{\prime} \mathbf{\Omega}^{-1}  \left( y_i - x_i \beta \right), \nonumber \\
\end{equation*}
where $r \left( \beta \right) = y_{it} - x_{it}^{\prime} \beta$ denote the error terms. 

In order to construct asymptotically consistent, weighted versions of the estimation equations,  we next introduce some definitions and notations on weighted likelihood methodology.

Let $M_{\beta} = \lbrace m_{\beta} \left( \cdot; \sigma_{\nu} \right); \sigma_{\nu} \in \mathbb{R}^+ \rbrace$ denote a parametric family of distributions for the theoretical error terms $r_i(\beta)$. We define $f^*(\cdot)$, the kernel density estimator based on the empirical distribution $\widehat{F}_N$ of the observed values of the residuals $r_i(\widehat{\beta})$, $i = 1, \cdots, N$, and $m_{\beta}^*(\cdot;\cdot)$, the smoothed model density, for $r_i(\widehat{\beta})$ as follows.
\begin{eqnarray} \nonumber
\label{f-m}
f^* \left( r_i \left( \widehat{\beta} \right) \right) &=& \int k \left( r_i \left( \widehat{\beta} \right);t, h \right) d \widehat{F}_N\left( t \right),\ \mbox{and}\\
m_{\beta}^* \left( r_i \left( \widehat{\beta} \right); \widehat{\sigma}_{\nu}\right) &=& \int k \left( r_i \left( \widehat{\beta} \right); t, h \right) d M_{\beta}\left(t; \widehat{\sigma}_{\nu} \right)\nonumber,
\end{eqnarray}
where $M_{\beta} \left(\cdot ; \sigma_{\nu} \right)$ is the distribution function for density $m_{\beta}(\cdot; \sigma_{\nu})$ and $k \left( r; t, h \right)$ is a kernel density with bandwidth $h$. In this study, the normal kernel density with variance $h^2$, $ k \left( r; t, h \right) = \frac{exp \left( - \left( r-t \right)^2 / 2h^2 \right)}{\sqrt{2\pi}h}$, is used. Note that the bandwidth parameter $h$ is chosen as $h=c  \sigma_{\nu}$ where $c$ is a constant term independent of the scale of the model so that outlying points will receive very small weights (cf. \cite{Markatou1998}). For the normal model, choosing the smoothing parameter based on the parameter $c$ in determining the level of downweighting ensures that the weighted likelihood estimating equations become location and scale equivariant as noted in \cite{Markatou1998}. We then define the Pearson residuals as follows.

\begin{equation*} 
\delta \left(  r_i \left( \widehat{\beta} \right) \right) = \frac{f^* \left( r_i \left( \widehat{\beta} \right) \right)}{m_{\beta}^* \left( r_i \left( \widehat{\beta} \right); \widehat{\sigma}_{\nu}\right)}-1
\end{equation*}
Based on the above, the weighted likelihood estimators of $\beta$ and $\sigma_{\nu}$ are obtained by solving the following estimating equations.
\begin{eqnarray}
\sum_{i=1}^{N} \omega\left( r_i \left( \widehat{\beta} \right); M_{\beta}, \widehat{F}_N \right) s \left( r_i\left( \beta \right); \sigma_{\nu} \right) = 0, \label{eq:eqbeta} 
\\
\sum_{i=1}^{N} \omega\left( r_i \left( \widehat{\beta} \right); M_{\beta}, \widehat{F}_N \right) s_{\sigma_{\nu}}\left( r_i\left( \beta \right); \sigma_{\nu} \right) = 0, \label{eq:eqsigma}
\end{eqnarray}
where 
\begin{eqnarray*}
s \left( r_i \left( \beta \right); \sigma_{\nu} \right) &=& \frac{\partial}{\partial \beta} \log{f \left( y_{it}; x_{it}, \beta, \sigma_{\nu} \right)} = \frac{\partial}{\partial \beta} \log{m_{\beta}(r_i\left( \beta \right); \sigma_{\nu})},\ \mbox{and} \\
s_{\sigma_{\nu}} \left( r_i \left( \beta \right); \sigma_{\nu} \right) &=& \frac{\partial}{\partial \sigma_{\nu}} \log{f \left( y_{it}; x_{it}, \beta, \sigma_{\nu} \right)} = \frac{\partial}{\partial \sigma_{\nu}} \log{m_{\beta}(r_i\left( \beta \right); \sigma_{\nu})}
\end{eqnarray*}
are the usual score functions and
\begin{equation*} 
\omega\left( r_i \left( \widehat{\beta} \right); M_{\beta}, \widehat{F}_N \right) = \omega_{i} = \min{ \left\lbrace 1, \frac{\left[ A\left( \delta \left(  r_i \left( \widehat{\beta} \right) \right)\right) + 1 \right]^+}{\delta \left(  r_i \left( \widehat{\beta} \right) \right)+ 1}\right\rbrace} 
\end{equation*}
where $ \left[~ . ~\right]^+ $ and $ A\left( . \right) $ denote the positive part of a function and the Residual Adjustment Function (RAF) as described in \cite{Lindsay1994} (e.g., Hellinger RAF $A(\delta) = 2 \left[ \left( \delta + 1 \right)^{1/2} - 1 \right]$), respectively. When $A\left( \delta \left(  r_i \left( \widehat{\beta} \right) \right)\right) = \delta \left(  r_i \left( \widehat{\beta} \right) \right)$, the weights $\omega\left( r_i \left( \widehat{\beta} \right); M_{\beta}, \widehat{F}_N \right) = 1$, and this leads to produce maximum likelihood estimates of the parameters (cf. \cite{Agostinelli2002b} and \cite{Agostinelli2001}).

In weighted likelihood methodology, the usual score equations based on maximum likelihood model are replaced by the weighted score equations to estimate model parameters. The weighted score equations defined above use the weights expressed as a function of Pearson residuals, $\delta \left(  r_i \left( \widehat{\beta} \right) \right) = \frac{f^* \left( r_i \left( \widehat{\beta} \right) \right)}{m_{\beta}^* \left( r_i \left( \widehat{\beta} \right); \widehat{\sigma}_{\nu}\right)}-1$. The weight function reflects the discordance between assumed model density and an estimate of true model density as noted in \cite{Markatou1996}. If the model is correctly specified in the absence outlying observations, then  $\delta$ converges with probability $1$ to $0$  and thus, the weight function assigns a value close to $1$. However, if the data involve outlying observations, large Pearson residuals are produced and the weight function assigns small weights to the outlying points depending on the level of discordance between the kernel density estimate of the model $f^* \left( \cdot \right)$ and the smoothed model density $m_{\beta}^* \left( \cdot; \cdot \right)$. Thus, the proposed estimators obtained using weighted likelihood estimating equations defined in \ref{eq:eqbeta} and \ref{eq:eqsigma} will be robust in presence of outliers and/or contamination in the data due to use of weighted residuals. 

An algorithm using resampling techniques have been proposed by \cite{Markatou1998} to find the roots of the weighted likelihood estimating equations. They suggest to use of data-driven starting values to create a reasonable search region that includes all reasonable solutions having high probability in parameter space. To this end, the sub-samples with fixed dimension, which are sufficiently large for obtaining the ML estimates of parameters $\beta$, are drawn without replacement from the data. Then, the ML estimates of $\beta$, $\widehat{\beta}_b^*$ for $b = 1, \cdots, B$ are obtained for each bootstrap sample. Finally, each of these estimates is used as an initial value in the iterative re-weighting algorithm for obtaining the roots of weighted likelihood estimating equations. $B = 30$ bootstrap sub-samples are created, and the maximum number of iterations are determined as $500$ in our simulation studies. (as in the default values of \texttt{R} package \texttt{wle})

The ML method can be considered as a minimum distance (minimum disparity) method and growing attention has been paid to construct a parallel method of estimation which has the similar or same efficiency properties with the ML method until the late 1970s. \cite{Beran1977} has focused the robustness properties of density based minimum distance estimation methods and demonstrated asymptotic first order efficiency of the estimator which minimizes the Hellinger distance between a kernel density estimator and a density from the model family within the continuous parametric models framework. The robustness of our proposed estimators is based on using the parallel minimum disparity measure in obtaining weight function for which downweight the outlying observations in the data. One of the main advantages related to the robustness properties in the minimum disparity estimation is that the presence of the valid objective function allows to investigate the breakdown point of the estimates as a measure of the robust global property. The breakdown properties of the estimators based on weighted likelihood estimating equations are examined by \cite{Markatou1996} and \cite{Markatou1998} using the stability property of the estimating equations. The root selection method plays a very crucial role in determining the theoretical breakdown properties of the estimators when an estimating equation has multiple roots (cf. \cite{Markatou1998}). To achieve the robust global property, a root is chosen based on using minimum parallel disparity measure defined as follows
\begin{equation} \nonumber
\rho_G \left( f^*, m_{\beta}^* \right) = \int G \left( \delta \left( x \right) \right) m_{\beta}^* \left( x \right) dx
\end{equation}
where $G$ is a thrice differentiable convex function defined on $\left[ -1, \infty \right)$ with $G \left( 0 \right) = 0$. \cite{Lindsay1994} has indicated that the choice of RAF may have a great impact on the robustness and efficiency of the corresponding estimators in the class of minimum disparity type methods. The function $G \left( \delta \right) = 2 \left( \left( \delta + 1 \right)^{1/2}-1 \right)^2$ is the squared Hellinger distance in our proposed approach. Under differentiability and regularity conditions, $\widehat{\beta}_i$ for $i = 1, 2, \cdots, \ell$ is obtained as a root of the minimum disparity estimating equation
\begin{equation} \nonumber
\int A \left( \delta (x) \right) \nabla m_{\beta}^* \left( x \right) dx = 0
\end{equation}
where $\nabla$ denote the gradient with respect to $\beta$, $A \left( \delta \right) = G^{\prime} \left( \delta \right) (1 + \delta) - G \left( \delta \right) $, $G^{\prime}$ representing the derivative of $G$, and $\delta (x)+ 1 = f^* \left( x \right) / m_{\beta}^* \left( x \right)$.
The parallel disparity measures obtained for each $\widehat{\beta}_i$ where $i = 1, 2, \cdots, \ell$, $\rho_G \left( f^*, m_{\widehat{\beta_1}}^* \left( x \right) \right), \rho_G \left( f^*, m_{\widehat{\beta_2}}^*\left( x \right) \right), \cdots, \rho_G \left( f^*, m_{\widehat{\beta_{\ell}}}^*\left( x \right)\right)$ are examined. Then, the proposed estimators based on weighted likelihood estimating equations can achieve the highest asymptotic breakdown point of $1/2$ by selecting a root providing the minimum value of disparity measure as shown below.
\begin{equation} \nonumber
\widehat{\beta}_{\omega} = \underset{\widehat{\beta}_i, i=1, 2, \cdots, l}{\mathrm{arg~min}} \rho_G \left( f^*, m_{\widehat{\beta}_i}^* \left( x \right) \right)
\end{equation}
Let $\beta_0$, $\widehat{\beta} =  \underset{\beta \in \mathbb{R}^{K}}{\mathrm{arg~min}} \sum_{i=1}^N r_i^2 \left( \beta \right)$ and $\widehat{\beta}_{\omega} = \underset{\beta \in \mathbb{R}^{K}}{\mathrm{arg~min}} \sum_{i=1}^N \omega_i r_i^2 \left( \beta \right)$ denote the true value of the parameters, the ML (or GLS) and weighted likelihood estimators of the parameters, respectively. For the linear panel data model with random effects, the conditions required for the existence of solutions and asymptotic normality of the proposed estimators are as follows (cf. \cite{Agostinelli1998} and \cite{Markatou1998}):
\begin{itemize}
\item[A1.] The weight function $\omega\left( \delta \right)$ is a nonnegative, bounded and differentiable function with respect to $\delta$. 
\item[A2.] The weight function $\omega\left( \delta \right)$ is regular with bounded $\omega^{\prime} \left( \delta \right) \left( \delta + 1 \right)$, where prime denotes the derivative. 
\end{itemize}

Let $\tilde{s}\left(x; \beta\right) = \nabla m_{\beta}^* \left( x \right) / m_{\beta}^* \left( x \right)$ and $s\left(x; \beta\right) = \nabla m_{\beta} \left( x \right) / m_{\beta} \left( x \right)$ where $m_{\beta}^* \left( x \right)$ and $m_{\beta} \left( x \right)$ denote the smoothed model and true model, respectively.
\begin{itemize}
\item[A3.] For every $\beta_0 \in \Theta$, there is a neighborhood $N \left( \beta_0 \right)$ such that for $\beta \in N \left( \beta_0 \right)$,  $M_i \left( x \right)$ for $i=1,2,3,4$, where $E_{\beta_0} \left[ M_i \left( X \right) \right] < \infty$, are the bounds for the quantities $\vert \tilde{s}\left(x; \beta\right) s^{\prime}\left(x; \beta\right) \vert$, $\vert \tilde{s}^2\left(x; \beta\right) s\left(x; \beta\right) \vert$, $\vert \tilde{s}^{\prime}\left(x; \beta\right) s\left(x; \beta\right) \vert$ and $\vert s^{\prime \prime}\left(x; \beta\right) \vert$.
\item[A4.] $E_{\beta_0} \left[ \tilde{s}^2\left(x; \beta\right) s^2\left(x; \beta\right) \right] < \infty$.
\item[A5.] The Fisher information is finite; $I \left( \beta \right) = E_{\beta} \left[ s^2\left(x; \beta\right) \right] < \infty$.
\item[A6.] \begin{itemize}
\item[i.] $\int \vert \nabla m_{\beta} \left( x \right) / m_{\beta}^* \left( x \right)\vert dx = \int \vert  m_{\beta} \left( x \right) s\left(x; \beta\right) / m_{\beta}^* \left( x \right)\vert dx < \infty$.
\item[ii.] $\int \vert \tilde{s}\left(x; \beta\right) s\left(x; \beta\right) \vert \left[ \frac{m_{\beta} \left( x \right)}{m_{\beta}^* \left( x \right)} \right] dx < \infty$.
\item[iii.] $\int \vert s^{\prime}\left(x; \beta\right)\vert \left[ \frac{m_{\beta} \left( x \right)}{m_{\beta}^* \left( x \right)} \right] dx < \infty$.
\end{itemize}
\item[A7.] The kernel density function $k \left( X; t, h \right)$ is bounded for all $x$ by a finite constant $M(h)$ that may depend on the smoothing parameter $h$ but not on $t$ or $x$.
\end{itemize}
Also, we present the following lemma from \cite{Markatou1995} needed for completeness.
\begin{lemma}
\label{lm:1}
\begin{eqnarray} \nonumber
\sqrt{N} \left| \frac{1}{N} \sum_{i=1}^N \omega \left( \delta \left( r_i \right) \right) s \left( x_i, y_i; \beta_0 \right) - \frac{1}{N} \sum_{i=1}^N s \left( x_i, y_i; \beta_0 \right) \right| \xrightarrow{p} 0 ~~as~~ N \rightarrow \infty \\ 
\left| \frac{1}{N} \sum_{i=1}^N \nabla_{\beta} \left\lbrace \omega \left( \delta \left( r_i \right) \right) s \left( x_i, y_i; \beta \right) \right\rbrace \vert_{\beta=\beta_0} - \frac{1}{N} \sum_{i=1}^N \frac{\partial}{\partial \beta} s \left( x_i, y_i; \beta_0 \right)  \right| \xrightarrow{p} 0 ~~as~~ N \rightarrow \infty \nonumber
\end{eqnarray}
and 
\begin{equation}
\frac{1}{N} \sum_{i=1}^N \frac{\partial^2}{\partial \beta^2} \left( \omega \left( \delta \left( r_i \right) \right) s \left( x_i, y_i; \beta \right) \right) \vert_{\beta=\beta^*} = \mathcal{O}_p \left( 1 \right) \nonumber
\end{equation}
where $\beta^*$ is the initial value between true value $\beta_0$ and $\widehat{\beta}_{\omega}$.
\end{lemma}
Under the assumption that the model is correctly specified and the conditions A1-A7 given above hold, we present the following theorem which shows the asymptotic equivalence of $\widehat{\beta}_{\omega}$ and $\widehat{\beta}$.

\begin{theorem}
\label{th:1}
\begin{equation}
\sqrt{N} \left( \widehat{\beta}_{\omega} - \widehat{\beta} \right) = o_p \left( 1 \right) ~~ as ~~ N \rightarrow \infty. \nonumber
\end{equation}
\end{theorem}

\begin{proof}
Using the fact that $\sup_{i} \vert \hat{\omega}_i - 1\vert \xrightarrow{p} 0$ (cf. \cite{Agostinelli2002b}) and Lemma \ref{lm:1} and observing that $\int \omega \left( \delta \left( r \right) \right) \left( \nabla_{\beta} s \left( x, y; \beta \right) \right) dF\left( x,y \right)$ is a continuous function in $\beta$, the proof of Theorem \ref{th:1} follows from Theorem 1 of \cite{Agostinelli2002b}.
\end{proof}

\section{Numerical Results} \label{Sec:3}
In this section, we present results from an extensive simulation study to assess the finite sample properties of the proposed and conventional estimators. The robustness performances of the proposed procedures are examined via three different scenarios; (i) different sample sizes, (ii) different error distributions, and (iii) various types of outliers. All calculations have been carried out using \texttt{R} 3.6.0. on an IntelCore i7 6700HQ 2.6 GHz PC. (The codes can be obtained from the author upon request.)

The following static linear panel data model is considered for the data generation processes (DGP).
\begin{equation*} \label{Eq:dgp}
y_{it} = x_{it}^{\prime} \beta + \alpha_i + \varepsilon_{it},~~i = 1, \ldots, N,~~t = 1, \ldots, T ,
\end{equation*} 
where $\alpha_i$ and $\varepsilon_{it}$'s are assumed to be iid $\text{N}(\mu = 0, \sigma^2 = 1)$. The vector of regression coefficients is chosen as $\beta^{\prime} = \left( \beta_1, \beta_2 \right) = \left( 2.4, -1.2 \right)$. For the random effects model specification (DGP-II), the explanatory variables $x_{itK}$ for $K = 1, 2$ are generated from a standard normal distribution. For the fixed effect model (DGP-I), these are generated depending on the individual effects as in \cite{Visek2015} as follows.
\begin{equation*}
x_{itK}^{fe} = x_{itK} + \alpha_i.
\end{equation*}
Throughout the experiments, $S = 1000$ simulations are performed to estimate the model coefficients and calculate the performance metrics. To evaluate the performance of the methods previously described, we calculate the mean squared errors (MSE):
$MSE = \frac{1}{S} \sum_{s=1}^S \Big\| \widehat{\beta}^s-\beta \Big\|^2$, where $\widehat{\beta}^s$, $s=1, \cdots, S$, denote the estimates obtained from $S$ simulated samples. The bandwidth of the kernel $k(\cdot)$ in Eq.~ \ref{f-m} is chosen using the \texttt{wle.smooth} function in the \texttt{R} package \texttt{wle}. 
\subsection{Sample sizes} \label{Sec:3.1}
Different values of cross-sectional dimension $N$ and time dimension $T$ are considered to investigate the effect of panel sizes on the performances of our proposed estimators. In particular, we consider $N = 25, 50, 100, 250$ for fixed time period $T = 4$, and $T = 3, 8, 12, 25$ for fixed cross-sectional dimension $N = 50$. We compare the MSE values of the estimators under standard normal errors, $\varepsilon_{it} \sim$ $\text{N}(\mu = 0, \sigma^2 = 1)$. The simulation results are presented in Table \ref{Tab1:N_increasing_T_increasing}. Our records indicate that, for both DGPs, the proposed weighted likelihood based estimators perform similarly with their traditional counterparts. These results also confirm that the proposed methods are consistent with the original least squares based estimators when $N$ and/or $T$ goes to infinity. 
\subsection{Error distributions} \label{Sec:3.2}
Three different error distributions, namely, $\text{N}(\mu = 0, \sigma^2 = 1)$, Student's t distribution with 5 degrees of freedom ($t_5$), and double exponential distribution with rate 1 ($DExp(1)$) are considered to evaluate the influence of error distributions on the estimation methods. The MSE values of the estimators are calculated for three pairs of values of cross-sectional sizes and time periods: $(N, T)=(250,5), (100, 10)$, and $(30, 20)$. Since our conclusions do not vary significantly with different choices of panel sizes, therefore to save space, we report only the results for  $N = 100$ and $T =10$. The results are reported in Table \ref{Tab2:error-dist-N-100-T-10} which indicate that the estimators have similar performances under different error distributions. 
\subsection{Outliers} \label{Sec:3.3}
In this section, the finite sample properties of the estimation procedures are investigated in the presence of different types of outliers. Throughout the simulations, the panel size is chosen to be $240$ with two levels of cross-sectional sizes and time periods, namely, $N_1 = 120$, $T_1 = 2$ and $N_2 = 80$, $T_2 = 3$. Two different levels of contamination (5\% and 10\%) are considered by setting the number of outliers as $m = 12$ and $m = 24$. The contaminated data is generated by two different ways as in \cite{Bramati2007} and \cite{Aquaro2013}: (i) outliers are randomly allocated over all observations, and (ii) half of the observations within individual units are contaminated such that outlying observations are concentrated in some blocks. Based on the above framework, the following contamination schemes are considered. 
\begin{itemize}
\item[1.] Random vertical outliers are generated by multiplying the randomly selected original values of the response variable by $-3$.
\item[2.] To generate the random leverage points, first the randomly selected values of the response variable are contaminated by multiplying with $-3$ and adding a $\text{N}(\mu = 20, \sigma^2 = 4)$ term. Then, the values of the explanatory variables corresponding to the contaminated values of the response variable are generated from a normal distribution $\text{N}(\mu = 5, \sigma^2 = 4)$. 
\item[3.] Concentrated vertical outliers are generated by multiplying the randomly selected blocks of the original values of response variable by $-3$ and adding a random value from $\text{N}(\mu = 50, \sigma^2 = 1)$. 
\item[4.] Concentrated leverage points are inserted into the randomly selected blocks of the original values of response and corresponding explanatory variables following the same rule as in the second scheme. 
\end{itemize}

Note that the proportion of contaminated values per cross-sectional unit constitute at least a half of observations over time periods for concentrated vertical outliers and concentrated leverage points. For the representative plots of the contamination schemes mentioned above, please see Figure 1 of \cite{Bramati2007}. The simulation results are given in Table \ref{Tab3:outliers-N-120-80-T-2-3-mse}. Our results clearly demonstrate that the proposed estimators outperform the conventional least-square based estimators in all situations. Note that, the performances of the conventional estimators can be severely degraded in the presence of outliers depending on the types and levels of contaminations. On the other hand, the performances of the proposed estimators are not sensitive to the choice of contamination level and/or scheme. We further see that, compared to the presence of vertical outliers, the traditional estimators produce more biased and less efficient results when the data is contaminated by the leverage points, in general. The MSE values calculated for traditional estimators significantly increase with increasing level of contamination in the presence of vertical outliers under DGP-II. It can further be seen that the proposed procedures ($\widehat{\beta}_{wpols}$ and $\widehat{\beta}_{wbe}$) produce significantly better results compared to $\widehat{\beta}_{pols}$ and $\widehat{\beta}_{be}$ estimators under all contamination schemes and contamination levels. The $\widehat{\beta}_{wre}$ and $\widehat{\beta}_{wfe}$ estimators exhibit improved performances over the conventional counterparts especially for DGP-II.  

To justify the superiority of the proposed methods further, we compare the power of significance tests of the regression coefficients. For the comparisons, the significance level $\gamma$ is set to $0.05$ to calculate the power of significance testing of individual coefficients which is defined as follows.
\begin{equation*}
P_K = \frac{1}{S} \sum_{s=1}^S \mathbf{1} \left\lbrace Q\left(\gamma/2 \right) \leq \frac{\widehat{\beta}_K^s}{\widehat{s}_K^s} \leq Q\left(1 - \gamma/2 \right) \right\rbrace,
\end{equation*}
where $\mathbf{1}(\cdot)$ and $Q(\cdot)$ denote the indicator function and the quantiles of standard normal distribution, respectively, and $\widehat{s}_K$ represents the standard error of the estimate for $K = 1, 2$. \\

Table \ref{Tab4:outliers-power} report the simulated powers for the individual coefficients $\beta_1 = 2.4$ and $\beta_2 = -1.2$, respectively. The results demonstrate that both proposed and traditional methods produce similar powers for $\beta_1$ when the data are contaminated by random vertical outliers. On the other hand, the proposed estimators have significantly better power values than those of classical methods for $\beta_2$. Additionally, while $\widehat{\beta}_{1,wfe}$ and $\widehat{\beta}_{1,wre}$ have similar power results when the data include concentrated vertical outliers, the proposed methods outperforms the unweighted procedures for $\beta_2$. The proposed methods are less affected by increasing number of outliers and yield a large gain in power in almost all cases. \\

\section{Case Study} \label{Sec:4}
In this section, we study the performances of the proposed and traditional OLS based estimators with a case study, the blood pressure data
set. The data set, which is consisted of a total of 2400 observations ($N = 200$, $T = 12$) with three variables: pulse rate, systolic pressure and diastolic pressure, are collected from male and female patients (during hospitalization) by ambulatory blood pressure monitors.  Let $x_{it1}$ = systolic pressure, $x_{it2}$ = diastolic pressure and $y_{it}$ = pulse rate, we conduct a linear panel data regression model: $ y_{it} = x_{it}^{\prime} \beta + \alpha_i + \varepsilon_{it} $ where $\beta^{\prime} = \left( \beta_1, \beta_2 \right)$, $i = 1, \ldots, 200$ and $t = 1, \ldots, 12$. The scatterplots of the response variable against the explanatory variables are presented in Figure \ref{Fig:1}. It is evident form the scatterplots that both data sets include outlying observations and the number of outliers in blood pressure data for females seems larger than those of males. The estimates of individual coefficients and standard errors of the estimates obtained for this dataset are shown in Table \ref{Tab5:female-male}. For both data sets gathered from 200 male and 200 female patients, our proposed weighted procedures yield slightly more efficient estimates than the OLS method, in general. The proposed $\widehat{\beta}_{wbe}$ have significantly better performances for estimating parameters compared to its traditional version. 

Moreover, we compared the predictive performances of the OLS and proposed methods. In doing so, the datasets are divided into the following two parts; the model is constructed based on the randomly selected 150 male and female patients, and the pulse rates of the remaining 50 patients are predicted using the estimated model parameters. This process is repeated 100 times, and for each time, the MSE for the predicted and observed pulse rates, $MSE = \frac{1}{50 \times 12} \sum_{i=1}^{50} \sum_{t=1}^{12} (y_{it} - \hat{y}_{it})^2$, are computed. The results are presented in Figure~\ref{Fig:pred}. This figure shows that both methods have similar MSE values. This is due to fact that the number of outliers is relatively small for the datasets with sample size $N \times T = 2400$.

\section{Conclusions} \label{Sec:5}
In this paper, we propose asymptotically valid, robust estimation procedures to obtain parameter estimates in linear panel data models with fixed and random effects. The proposed approaches are based on using weighted likelihood methodology. The finite sample performances of the proposed methods are examined through extensive simulation studies and a real-world example, and the results are compared with existing methods. Our records show that the proposed procedures have similar performance with existing estimation methods when the data have no outliers and/or under different error distributions. However, our proposed method produces more accurate and efficient parameter estimates with better power compared to the traditional OLS methods when outliers are presented in the data. As a part of our future research, we will investigate whether the proposed methods can also be used to estimate the parameters in linear dynamic panel data models as an alternative to the generalized method of moments.

\clearpage

\begin{table}
\centering
\caption{The MSEs of all estimators under $\text{N}(\mu = 0, \sigma^2 = 1)$ errors for $N = 25, 50, 100, 250$ (when $T$ is fixed, $T = 4$) and $T = 3, 8, 12, 25$ (when $N$ is fixed, $N = 50$)}
\scriptsize
\begin{tabular}{l c c c c| c c c c}
\hline 
& \multicolumn{4}{c|}{DGP-I: Fixed Effects Model} & \multicolumn{4}{c}{DGP-II: Random Effects Model}  \\
($N$, $T$) & (25, 4) & (50, 4) & (100, 4) & (250, 4) & (25, 4) & (50, 4) & (100, 4) & (250, 4) \\
\hline
$\widehat{\beta}_{pols}$ & 0.2290 & 0.2255 & 0.2232 & 0.2222 &  0.0436 & 0.0200 & 0.0097 & 0.0041 \\
$\widehat{\beta}_{wpols}$ & 0.2295 & 0.2255 & 0.2232 & 0.2221 & 0.0441 & 0.0201 & 0.0097 & 0.0041 \\
$\widehat{\beta}_{be}$ & 0.2727 & 0.2388 & 0.2298 & 0.2265 & 0.1801 & 0.0864 & 0.0401 & 0.0162 \\
$\widehat{\beta}_{wbe}$ & 0.2780 & 0.2396 & 0.2300 & 0.2267 & 0.1912 & 0.0886 & 0.0411 & 0.0163 \\
$\widehat{\beta}_{fe}$  & 0.2385 & 0.2320 & 0.2257 & 0.2227 & 0.0600 & 0.0262 & 0.0126 & 0.0055 \\
$\widehat{\beta}_{wfe}$ & 0.2386 & 0.2321 & 0.2257 & 0.2227 & 0.0607 & 0.0263 & 0.0126 & 0.0055 \\
$\widehat{\beta}_{re}$ & 0.2291 & 0.2257 & 0.2231 & 0.2221 & 0.0438 & 0.0199 & 0.0097 & 0.0041 \\
$\widehat{\beta}_{wre}$ & 0.2296 & 0.2257 & 0.2231 & 0.2221 & 0.0443 & 0.0200 & 0.0097 & 0.0042 \\
\hline
($N$, $T$) & (50, 3) & (50, 8) & (50, 12) & (50, 25) & (50, 3) & (50, 8) & (50, 12) & (50, 25) \\
\hline
$\widehat{\beta}_{pols}$ & 0.2302 & 0.2206 & 0.2189 & 0.2193 & 0.0260 & 0.0103 & 0.0064 & 0.0030 \\
$\widehat{\beta}_{wpols}$ & 0.2302 & 0.2207 & 0.2189 & 0.2193 & 0.0261 & 0.0104 & 0.0064 & 0.0030 \\
$\widehat{\beta}_{be}$ & 0.2476 & 0.2190 & 0.2067 & 0.1417 & 0.0791 & 0.0793 & 0.0764 & 0.0673 \\
$\widehat{\beta}_{wbe}$ & 0.2474 & 0.2207 & 0.2076 & 0.1432 & 0.0816 & 0.0822 & 0.0781 & 0.0685 \\
$\widehat{\beta}_{fe}$  & 0.2397 & 0.2243 & 0.2217 & 0.2220 & 0.0395 & 0.0119 & 0.0072 & 0.0032 \\
$\widehat{\beta}_{wfe}$ & 0.2399 & 0.2243 & 0.2217 & 0.2221 & 0.0397 & 0.0119 & 0.0072 & 0.0032 \\
$\widehat{\beta}_{re}$ & 0.2307 & 0.2207 & 0.2190 & 0.2193 & 0.0260 & 0.0103 & 0.0064 & 0.0030 \\
$\widehat{\beta}_{wre}$ & 0.2307 & 0.2208 & 0.2189 & 0.2193 & 0.0261 & 0.0103 & 0.0064 & 0.0030 \\
\hline
\end{tabular}
\label{Tab1:N_increasing_T_increasing}
\end{table}

\begin{table}
\centering
\caption{The MSEs of all estimators under different error distributions: $\text{N}(\mu = 0, \sigma^2 = 1)$, $t_5$ and $DExp(1)$ when $N = 100$, $T = 10$}
\scriptsize
\begin{tabular}{l c c c| c c c}
\hline 
Error & \multicolumn{3}{c|}{DGP-I: Fixed Effects Model} & \multicolumn{3}{c}{DGP-II: Random Effects Model}  \\
distributions & $\text{N}(\mu = 0, \sigma^2 = 1)$ & $t_5$ & $DExp(1)$ & $\text{N}(\mu = 0, \sigma^2 = 1)$ & $t_5$ & $DExp(1)$ \\ 
\hline
$\widehat{\beta}_{pols}$ & 0.2203 & 0.2193 & 0.2251 & 0.0037 & 0.0053 & 0.0057 \\
$\widehat{\beta}_{wpols}$ & 0.2203 & 0.2191 & 0.2250 & 0.0037 & 0.0051 & 0.0054 \\
$\widehat{\beta}_{be}$ & 0.2139 & 0.2308 & 0.2346 & 0.0373 & 0.0537 & 0.0601 \\
$\widehat{\beta}_{wbe}$ & 0.2146 & 0.2306 & 0.2358 & 0.0381 & 0.0540 & 0.0601 \\
$\widehat{\beta}_{fe}$ & 0.2215 & 0.2200 & 0.2263 & 0.0042 & 0.0062 & 0.0064 \\
$\widehat{\beta}_{wfe}$ & 0.2215 & 0.2199 & 0.2263 & 0.0042 & 0.0059 & 0.0061 \\
$\widehat{\beta}_{re}$ & 0.2203 & 0.2193 & 0.2251 & 0.0037 & 0.0054 & 0.0057 \\
$\widehat{\beta}_{wre}$ & 0.2203 & 0.2191 & 0.2250 & 0.0037 & 0.0051 & 0.0054 \\
\hline
\end{tabular}
\label{Tab2:error-dist-N-100-T-10}
\end{table}

\begin{table}
\centering
\caption{The MSEs of all estimators for $N_1 = 120$, $T_1 = 2$ and $N_2 = 80$, $T_2 = 3$ in the presence of 5\% and 10\% random and concentrated contamination by setting the number of outliers as $m = 12$ and $m = 24$}
\scriptsize
\begin{tabular}{l c c c c c c c c}
\hline
DGP-I & \multicolumn{4}{c}{Random contamination} & \multicolumn{4}{c}{Concentrated contamination} \\
\cmidrule{1-9}
($N_1$, $T_1$) = (120, 2)  & \multicolumn{2}{c}{Vertical outliers} & \multicolumn{2}{c}{Leverage points} & \multicolumn{2}{c}{Vertical outliers} & \multicolumn{2}{c}{Leverage points} \\ 
& 5\% & 10\% & 5\% & 10\% & 5\% & 10\% & 5\% & 10\% \\
\hline
$\widehat{\beta}_{pols}$ & 0.4002 & 1.2417 & 3.3460 & 9.9541 & 1.0588 & 2.4242 & 3.3060 & 9.9526 \\
$\widehat{\beta}_{wpols}$ & 0.2274 & 0.2228 & 0.2260 & 0.2280 & 0.2313 & 0.2302 & 0.2277 & 0.2289 \\
$\widehat{\beta}_{be}$ & 0.4502 & 1.3096 & 3.4462 & 10.1161 & 3.2571 & 6.2801 & 3.4380 & 9.7797 \\
$\widehat{\beta}_{wbe}$ & 0.2352 & 0.2265 & 0.2370 & 0.4014 & 0.2408 & 0.2390 & 0.2348 & 0.2428 \\
$\widehat{\beta}_{fe}$ & 0.4660 & 1.3876 & 3.4474 & 10.0996 & 0.4596 &  1.3367 & 3.6551 & 10.4212 \\
$\widehat{\beta}_{wfe}$ & 0.2330 & 0.2278 & 0.2329 & 0.4529 & 0.2364 & 0.2346 & 0.2330 & 0.2385 \\
$\widehat{\beta}_{re}$ & 0.4020 & 1.2458 & 3.3482 & 9.9689 & 0.4378 & 1.3038 & 3.0432 & 9.6394 \\
$\widehat{\beta}_{wre}$ & 0.2268 & 0.2219 & 0.2242 & 0.2265 & 0.2365 & 0.2359 & 0.2276 & 0.2290 \\
\hline 
($N_2$, $T_2$) = (80, 3) \\
$\widehat{\beta}_{pols}$ & 0.3746 & 1.1813 & 3.4531 & 9.7358 & 1.1149 & 2.6081 & 3.3712 & 10.0874 \\
$\widehat{\beta}_{wpols}$ & 0.2228 & 0.2250 & 0.2251 & 0.2262 & 0.2304 & 0.2230 &  0.2237 & 0.2248 \\
$\widehat{\beta}_{be}$ & 0.4705 & 1.4074 & 3.7326 & 10.0835 & 6.9713 & 14.2696 & 3.4437 & 10.3329 \\
$\widehat{\beta}_{wbe}$ & 0.2403 & 0.3131 & 0.2467 & 1.8689 & 0.2449 & 0.2364 & 0.2384 & 0.2591 \\
$\widehat{\beta}_{fe}$ & 0.4165 & 1.2425 & 3.4809 & 9.7985 & 0.4392 & 1.3105 & 3.6801 & 10.0476 \\
$\widehat{\beta}_{wfe}$ & 0.2176 & 0.2155 & 0.2234 & 0.6387 & 0.2329 & 0.2250 & 0.2310 & 0.2317 \\
$\widehat{\beta}_{re}$ & 0.3764 & 1.1805 & 3.4555 & 9.7389 & 0.4352 & 1.2977 & 3.1797 & 9.7787 \\
$\widehat{\beta}_{wre}$ & 0.2222 & 0.2239 & 0.2235 & 0.2236 & 0.2337 & 0.2266 & 0.2238 & 0.2236 \\
\hline
DGP-II & \multicolumn{4}{c}{Random contamination} & \multicolumn{4}{c}{Concentrated contamination} \\
\cmidrule{1-9}
($N_1$, $T_1$) = (120, 2) & \multicolumn{2}{c}{Vertical outliers} & \multicolumn{2}{c}{Leverage points} & \multicolumn{2}{c}{Vertical outliers} & \multicolumn{2}{c}{Leverage points} \\ 
& 5\% & 10\% & 5\% & 10\% & 5\% & 10\% & 5\% & 10\% \\
\hline
$\widehat{\beta}_{pols}$ & 0.3787 & 1.3251 & 3.5885 & 10.2682 & 1.3436 & 3.1538 & 3.7369 & 9.9809 \\
$\widehat{\beta}_{wpols}$ & 0.0190 & 0.0219 & 0.0340 & 0.0646 & 0.0176 & 0.0188 & 0.0371 & 0.0614 \\
$\widehat{\beta}_{be}$ & 0.4505 & 1.4368 & 3.7069 & 10.3196 & 4.4839 & 9.2887 & 3.6873 & 9.7482 \\
$\widehat{\beta}_{wbe}$ & 0.0480 & 0.0858 & 0.0976 & 0.3368 & 0.0361 & 0.0387 & 0.0963 & 0.2151 \\
$\widehat{\beta}_{fe}$ & 0.4508 & 1.4608  & 3.6715 & 10.5080 & 0.4970 & 1.4487 & 4.4034 & 11.2726 \\
$\widehat{\beta}_{wfe}$ & 0.0483 & 0.0999  & 0.1078 & 0.5219 & 0.0386 & 0.0441 & 0.0400 & 0.0482 \\
$\widehat{\beta}_{re}$ & 0.3798 & 1.3272 & 3.5925 & 10.2741 & 0.4703 & 1.4108 & 3.4956 & 9.7057 \\
$\widehat{\beta}_{wre}$ & 0.0191 & 0.0225 & 0.0352 & 0.0686 & 0.0358 & 0.0373 & 0.0362 & 0.0618 \\
\hline 
($N_2$, $T_2$) = (80, 3) \\
$\widehat{\beta}_{pols}$ & 0.3784 & 1.2896 & 3.6960 & 10.2332 & 1.4453 & 3.3062 & 3.7121 & 10.2756 \\
$\widehat{\beta}_{wpols}$ & 0.0173 & 0.0214 & 0.0358 & 0.0686 & 0.0183 & 0.0192 & 0.0346 & 0.0584 \\
$\widehat{\beta}_{be}$ & 0.5166 & 1.5569 &  3.8817 & 10.5179 & 10.2610 & 20.0915 & 3.5058 & 9.7968 \\
$\widehat{\beta}_{wbe}$ & 0.0900 & 0.4135 &  0.2439 & 2.2413 & 0.0589 & 0.0569 & 0.2733 & 0.3970 \\
$\widehat{\beta}_{fe}$ & 0.4176 & 1.3365 &  3.7580 & 10.2875 & 0.4187 & 1.4138 & 4.1639 & 11.4456 \\
$\widehat{\beta}_{wfe}$ & 0.0421 & 0.1894 &  0.1248 & 0.9373 & 0.0278 & 0.0324 & 0.0292 & 0.0322 \\
$\widehat{\beta}_{re}$ & 0.3791 & 1.2919 &   3.7027 & 10.2412 & 0.4143 & 1.4053 & 3.5679 & 10.0512 \\
$\widehat{\beta}_{wre}$ & 0.0174 & 0.0218 &  0.0360 & 0.0715 & 0.0271 & 0.0310 & 0.0335 & 0.0570 \\
\hline
\end{tabular}
\label{Tab3:outliers-N-120-80-T-2-3-mse}
\end{table}

\begin{landscape} 
\begin{table}
\centering
\caption{The powers of all estimators when $N_1 = 120$, $T_1 = 2$ and $N_2 = 80$, $T_2 = 3$}
\scriptsize
\begin{tabular}{l c c c c c c c c l c c c c c c c c}
\hline 
DGP-I & \multicolumn{4}{c}{Random contamination} & \multicolumn{4}{c}{Concentrated contamination} & DGP-II & \multicolumn{4}{c}{Random contamination} & \multicolumn{4}{c}{Concentrated contamination} \\
\cmidrule{1-18}
($N_1$, $T_1$) & \multicolumn{2}{c}{Vertical outliers} & \multicolumn{2}{c}{Leverage points} & \multicolumn{2}{c}{Vertical outliers} & \multicolumn{2}{c}{Leverage points} & ($N_1$, $T_1$) & \multicolumn{2}{c}{Vertical outliers} & \multicolumn{2}{c}{Leverage points} & \multicolumn{2}{c}{Vertical outliers} & \multicolumn{2}{c}{Leverage points} \\ 
(120, 2) & 5\% & 10\% & 5\% & 10\% & 5\% & 10\% & 5\% & 10\% & (120, 2) & 5\% & 10\% & 5\% & 10\% & 5\% & 10\% & 5\% & 10\% \\
\hline
$\widehat{\beta}_{1,pols}$ & 1.000 & 0.996 & 0.887 & 0.368 & 0.908 & 0.499 & 0.887 & 0.358 & $\widehat{\beta}_{1,pols}$ & 1.000 & 0.988 & 0.769 & 0.410 & 0.747 & 0.313 & 0.752 & 0.394 \\
$\widehat{\beta}_{1,wpols}$ & 1.000 & 1.000 & 1.000 & 0.999 & 1.000 & 1.000 & 1.000 & 0.999 & $\widehat{\beta}_{1,wpols}$ & 1.000 & 1.000 & 0.997 & 0.996 & 1.000 & 1.000 & 0.999 & 0.999 \\
$\widehat{\beta}_{1,be}$ & 1.000 & 0.981 & 0.789 & 0.258 & 0.480 & 0.178 & 0.768 & 0.341 & $\widehat{\beta}_{1,be}$ & 1.000 & 0.946 & 0.652 &  0.269 & 0.327 & 0.139 & 0.696 & 0.350 \\
$\widehat{\beta}_{1,wbe}$ & 1.000 & 1.000 & 1.000 & 0.995 & 1.000 & 1.000 & 1.000 & 0.999 &  $\widehat{\beta}_{1,wbe}$ & 1.000 & 1.000 & 0.997 & 0.994 & 1.000 & 1.000 & 0.999 & 0.999 \\
$\widehat{\beta}_{1,fe}$ & 1.000 & 0.970 & 0.786 & 0.251 & 1.000 & 0.960 & 0.834 & 0.336 & $\widehat{\beta}_{1,fe}$ & 0.927 & 0.949 & 0.654 &   0.285 & 0.996 & 0.919 & 0.597 & 0.354 \\
$\widehat{\beta}_{1,wfe}$ & 1.000 & 1.000 & 1.000 & 0.985 & 1.000 & 1.000 & 1.000 & 0.999 & $\widehat{\beta}_{1,wfe}$ & 1.000 & 1.000 & 0.997 &  0.986 & 1.000 & 1.000 & 0.999 & 0.999 \\
$\widehat{\beta}_{1,re}$ & 1.000 & 0.996 & 0.888 & 0.366 & 1.000 & 0.966 & 0.908 & 0.365 & $\widehat{\beta}_{1,re}$ & 1.000 & 0.987 & 0.771 &   0.413 & 1.000 & 0.934 & 0.767 & 0.400 \\
$\widehat{\beta}_{1,wre}$ & 1.000 & 1.000 & 1.000 & 0.999 & 1.000 & 1.000 & 1.000 & 0.999 & $\widehat{\beta}_{1,wre}$ & 1.000 & 1.000 & 0.997 &   0.996 & 1.000 & 1.000 & 0.999 & 0.999 \\
\hline
$\widehat{\beta}_{2,pols}$ & 0.945 & 0.604 & 0.379 & 0.650 & 0.217 & 0.090 & 0.406 & 0.673 & $\widehat{\beta}_{2,pols}$ & 0.989 & 0.816 & 0.580 & 0.495 & 0.294 & 0.121 & 0.532 & 0.493 \\
$\widehat{\beta}_{2,wpols}$ & 1.000 & 1.000 & 1.000 & 0.999 & 1.000 & 1.000 & 1.000 & 0.999 & $\widehat{\beta}_{2,wpols}$ & 1.000 & 1.000 & 0.997 & 0.996 & 1.000 & 1.000 & 0.999 & 0.999 \\
$\widehat{\beta}_{2,be}$ & 0.784 & 0.387 & 0.265 & 0.480 & 0.107 & 0.059 & 0.445 & 0.441 & $\widehat{\beta}_{2,be}$ & 0.909 & 0.559 & 0.440 &  0.336 & 0.123 & 0.069 & 0.544 & 0.419 \\
$\widehat{\beta}_{2,wbe}$ & 1.000 & 1.000 & 1.000 & 0.995 & 1.000 & 1.000 & 1.000 & 0.999 & $\widehat{\beta}_{2,wbe}$ & 1.000 & 1.000 & 0.997 &  0.994 & 1.000 & 1.000 & 0.999 & 0.999 \\
$\widehat{\beta}_{2,fe}$ & 0.772 & 0.366 & 0.267 & 0.475 & 0.774 & 0.422 & 0.315 & 0.584 & $\widehat{\beta}_{2,fe}$ & 0.998 & 0.545 & 0.449 &   0.342 & 0.903 & 0.586 & 0.417 & 0.263 \\
$\widehat{\beta}_{2,wfe}$ & 1.000 & 1.000 & 1.000 & 0.986 & 1.000 & 1.000 & 1.000 & 0.999 & $\widehat{\beta}_{2,wfe}$ & 1.000 & 0.999 & 0.996 &  0.981 & 1.000 & 1.000 & 0.999 & 0.999 \\
$\widehat{\beta}_{2,re}$ & 0.944 & 0.602 & 0.380 & 0.648 & 0.792 & 0.431 & 0.406 & 0.661 & $\widehat{\beta}_{2,re}$ & 0.987 & 0.818 & 0.578 &  0.498 & 0.912 & 0.596 & 0.556 & 0.480 \\
$\widehat{\beta}_{2,wre}$ & 1.000 & 1.000 & 1.000 & 0.999 & 1.000 & 1.000 & 1.000 & 0.999 & $\widehat{\beta}_{2,wre}$ & 1.000 & 1.000 & 0.997 &  0.996 & 1.000 & 1.000 & 0.999 & 0.999 \\
\hline 
DGP-I & \multicolumn{4}{c}{Random contamination} & \multicolumn{4}{c}{Concentrated contamination} & DGP-II & \multicolumn{4}{c}{Random contamination} & \multicolumn{4}{c}{Concentrated contamination} \\
\cmidrule{1-18}
($N_2$, $T_2$) & \multicolumn{2}{c}{Vertical outliers} & \multicolumn{2}{c}{Leverage points} & \multicolumn{2}{c}{Vertical outliers} & \multicolumn{2}{c}{Leverage points} & ($N_2$, $T_2$) & \multicolumn{2}{c}{Vertical outliers} & \multicolumn{2}{c}{Leverage points} & \multicolumn{2}{c}{Vertical outliers} & \multicolumn{2}{c}{Leverage points} \\ 
(80, 3) & 5\% & 10\% & 5\% & 10\% & 5\% & 10\% & 5\% & 10\% & (80, 3) & 5\% & 10\% & 5\% & 10\% & 5\% & 10\% & 5\% & 10\% \\
\hline
$\widehat{\beta}_{1,pols}$ & 0.999 & 0.994 & 0.864 & 0.361 & 0.902 & 0.499 & 0.890 & 0.374 & $\widehat{\beta}_{1,pols}$ & 1.000 & 0.992 & 0.776 & 0.415 & 0.731 & 0.340 & 0.735 & 0.409 \\
$\widehat{\beta}_{1,wpols}$ & 0.999 & 0.998 & 0.999 & 0.996 & 1.000 & 1.000 & 1.000 & 1.000 & $\widehat{\beta}_{1,wpols}$ & 1.000 & 1.000 & 0.998 & 0.995 & 1.000 & 1.000 & 0.997 & 0.999 \\
$\widehat{\beta}_{1,be}$ & 0.998 & 0.939 & 0.693 & 0.196 & 0.237 & 0.137 & 0.719 & 0.397 & $\widehat{\beta}_{1,be}$ & 0.995 & 0.868 & 0.591 &  0.224 & 0.182 & 0.103 & 0.710 & 0.354 \\
$\widehat{\beta}_{1,wbe}$ & 0.999 & 0.998 & 0.999 & 0.926 & 1.000 & 1.000 & 1.000 & 1.000 & $\widehat{\beta}_{1,wbe}$ & 1.000 & 0.996 & 0.998 &  0.914 & 1.000 & 1.000 & 0.996 & 0.997 \\
$\widehat{\beta}_{1,fe}$ & 0.999 & 0.989 & 0.817 & 0.285 & 0.999 & 0.982 & 0.872 & 0.374 & $\widehat{\beta}_{1,fe}$ & 1.000 & 0.982 & 0.711 &  0.337 & 1.000 & 0.961 & 0.628 & 0.399 \\
$\widehat{\beta}_{1,wfe}$ & 0.999 & 0.998 & 0.999 & 0.976 & 1.000 & 1.000 & 1.000 & 1.000 & $\widehat{\beta}_{1,wfe}$ & 1.000 & 1.000 & 0.998 &  0.963 & 1.000 & 1.000 & 0.997 & 0.999 \\
$\widehat{\beta}_{1,re}$ & 0.999 & 0.994 & 0.862 & 0.361 & 1.000 & 0.984 & 0.906 & 0.384 & $\widehat{\beta}_{1,re}$ & 1.000 & 0.992 & 0.776 &  0.418 & 1.000 & 0.966 & 0.757 & 0.406 \\
$\widehat{\beta}_{1,wre}$ & 0.999 & 0.998 & 0.999 & 0.996 & 1.000 & 1.000 & 1.000 & 1.000 & $\widehat{\beta}_{1,wre}$ & 1.000 & 1.000 & 0.998 &  0.995 & 1.000 & 1.000 & 0.997 & 0.999 \\
\hline
$\widehat{\beta}_{2,pols}$ & 0.952 & 0.638 & 0.399 & 0.627 & 0.205 & 0.108 & 0.404 & 0.652 & $\widehat{\beta}_{2,pols}$ & 0.988 & 0.829 & 0.577 & 0.507 & 0.239 & 0.121 & 0.577 & 0.519 \\
$\widehat{\beta}_{2,wpols}$ & 0.999 & 0.998 & 0.999 & 0.996 & 1.000 & 1.000 & 1.000 & 1.000 & $\widehat{\beta}_{2,wpols}$ & 1.000 & 1.000 & 0.998 & 0.995 & 1.000 & 1.000 & 0.997 & 0.999 \\
$\widehat{\beta}_{2,be}$ & 0.672 & 0.308 & 0.217 & 0.354 & 0.079 & 0.059 & 0.503 & 0.423 & $\widehat{\beta}_{2,be}$ & 0.811 & 0.452 & 0.363 &  0.258 & 0.076 & 0.052 & 0.570 & 0.421 \\
$\widehat{\beta}_{2,wbe}$ & 0.999 & 0.972 & 0.999 & 0.903 & 1.000 & 1.000 & 1.000 & 0.997 & $\widehat{\beta}_{2,wbe}$ & 1.000 & 0.968 & 0.997 &  0.913 & 1.000 & 1.000 & 0.993 & 0.995 \\
$\widehat{\beta}_{2,fe}$ & 0.884 & 0.480 & 0.321 & 0.522 & 0.850 & 0.466 & 0.317 & 0.666 & $\widehat{\beta}_{2,fe}$ & 0.955 & 0.673 & 0.491 &  0.402 & 0.938 & 0.681 & 0.416 & 0.253 \\
$\widehat{\beta}_{2,wfe}$ & 0.999 & 0.998 & 0.999 & 0.937 & 1.000 & 1.000 & 1.000 & 1.000 & $\widehat{\beta}_{2,wfe}$ & 1.000 & 0.998 & 0.997 &  0.947 & 1.000 & 1.000 & 0.997 & 0.999 \\
$\widehat{\beta}_{2,re}$ & 0.953 & 0.641 & 0.396 & 0.632 & 0.853 & 0.470 & 0.412 & 0.641 & $\widehat{\beta}_{2,re}$ & 0.989 & 0.831 & 0.580 &  0.506 & 0.941 & 0.683 & 0.586 & 0.494 \\
$\widehat{\beta}_{2,wre}$ & 0.999 & 0.998 & 0.999 & 0.996 & 1.000 & 1.000 & 1.000 & 1.000 & $\widehat{\beta}_{2,wre}$ & 1.000 & 1.000 & 0.998 &  0.995 & 1.000 & 1.000 & 0.997 & 0.999 \\
\hline
\end{tabular}
\label{Tab4:outliers-power}
\end{table}
\end{landscape}

\begin{table}
\centering
\caption{The estimates of individual coefficients (upper rows) and standard errors of the estimates (lower rows) for blood pressure data of female and male patients}
\scriptsize
\begin{tabular}{l c l c | l c l c}
\cmidrule{1-8} 
\multicolumn{4}{c|}{Female patients} & \multicolumn{4}{c}{Male patients} \\ 
\cmidrule{1-8} 
$\widehat{\beta}_{1,pols}$ & -0.109 & $\widehat{\beta}_{2,pols}$ & 0.414 & $\widehat{\beta}_{1,pols}$ & -0.030 & $\widehat{\beta}_{2,pols}$ & 0.335 \\
 & (0.015) &  & (0.024) &  & (0.017) &  & (0.024) \\
$\widehat{\beta}_{1,wpols}$ & -0.114 & $\widehat{\beta}_{2,wpols}$ & 0.423 & $\widehat{\beta}_{1,wpols}$ & -0.022 & $\widehat{\beta}_{2,wpols}$ & 0.334 \\
 & (0.014) &  & (0.023) &  & (0.016) &  & (0.022) \\
$\widehat{\beta}_{1,be}$ & -0.250 & $\widehat{\beta}_{2,be}$ & 0.583 & $\widehat{\beta}_{1,be}$ & -0.118 & $\widehat{\beta}_{2,be}$ & 0.352 \\
 & (0.052) &  & (0.092) &  & (0.061) &  & (0.086) \\
$\widehat{\beta}_{1,wbe}$ & -0.255 & $\widehat{\beta}_{2,wbe}$ & 0.589 & $\widehat{\beta}_{1,wbe}$ & -0.116 & $\widehat{\beta}_{2,wbe}$ & 0.346  \\
 & (0.014) &  & (0.025) &  & (0.016) &  & (0.023) \\
$\widehat{\beta}_{1,fe}$ &  0.101 & $\widehat{\beta}_{2,fe}$ & 0.214 & $\widehat{\beta}_{1,fe}$ & 0.076 & $\widehat{\beta}_{2,fe}$ & 0.301 \\
 & (0.015) &  & (0.022) &  & (0.017) &  & (0.023) \\
$\widehat{\beta}_{1,wfe}$ & 0.098 & $\widehat{\beta}_{2,wfe}$ & 0.216 & $\widehat{\beta}_{1,wfe}$ & 0.081 & $\widehat{\beta}_{2,wfe}$ & 0.299 \\
 & (0.013) &  & (0.019) &  & (0.016) &  & (0.021) \\
$\widehat{\beta}_{1,re}$ & 0.070 & $\widehat{\beta}_{2,re}$ & 0.241 & $\widehat{\beta}_{1,re}$ & 0.061 & $\widehat{\beta}_{2,re}$ & 0.307 \\
 & (0.015) &  & (0.021) &  & (0.017) &  & (0.023) \\
$\widehat{\beta}_{1,wre}$ & 0.065 & $\widehat{\beta}_{2,wre}$ & 0.246 & $\widehat{\beta}_{1,wre}$ & 0.066 & $\widehat{\beta}_{2,wre}$ & 0.307 \\
 & (0.013) &  & (0.018) &  & (0.015) &  & (0.020) \\
\cmidrule{1-8}  
\end{tabular}
\label{Tab5:female-male} 
\end{table}

\clearpage
\begin{figure}[!htbp]
  \centering
  \includegraphics[width=7cm]{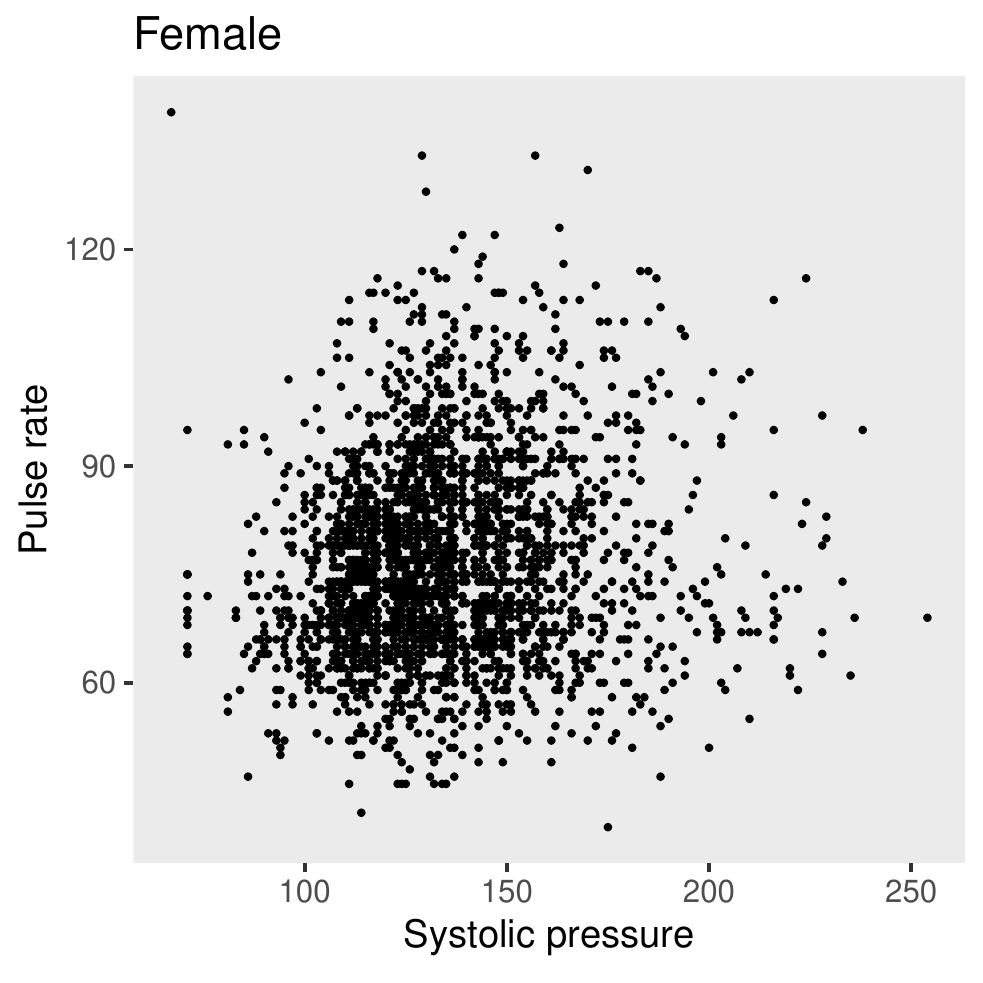}
  \includegraphics[width=7cm]{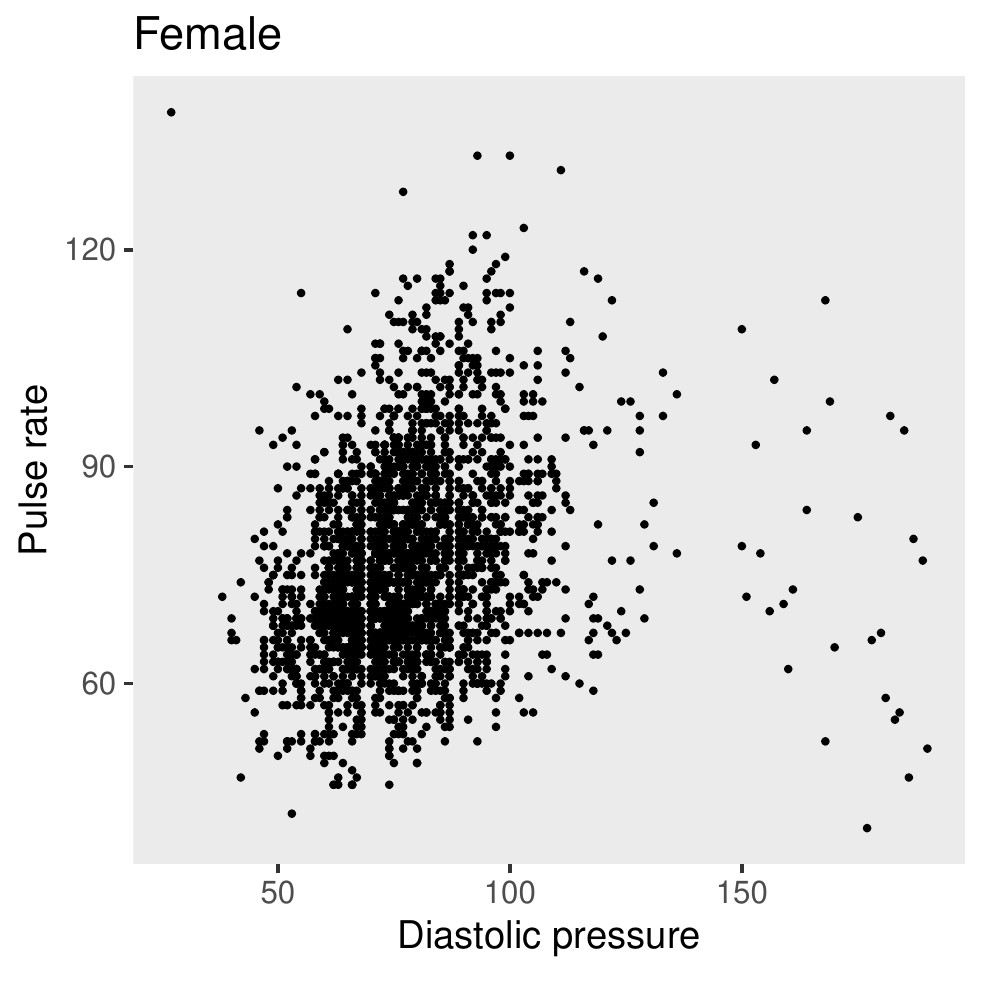}
\\
  \includegraphics[width=7cm]{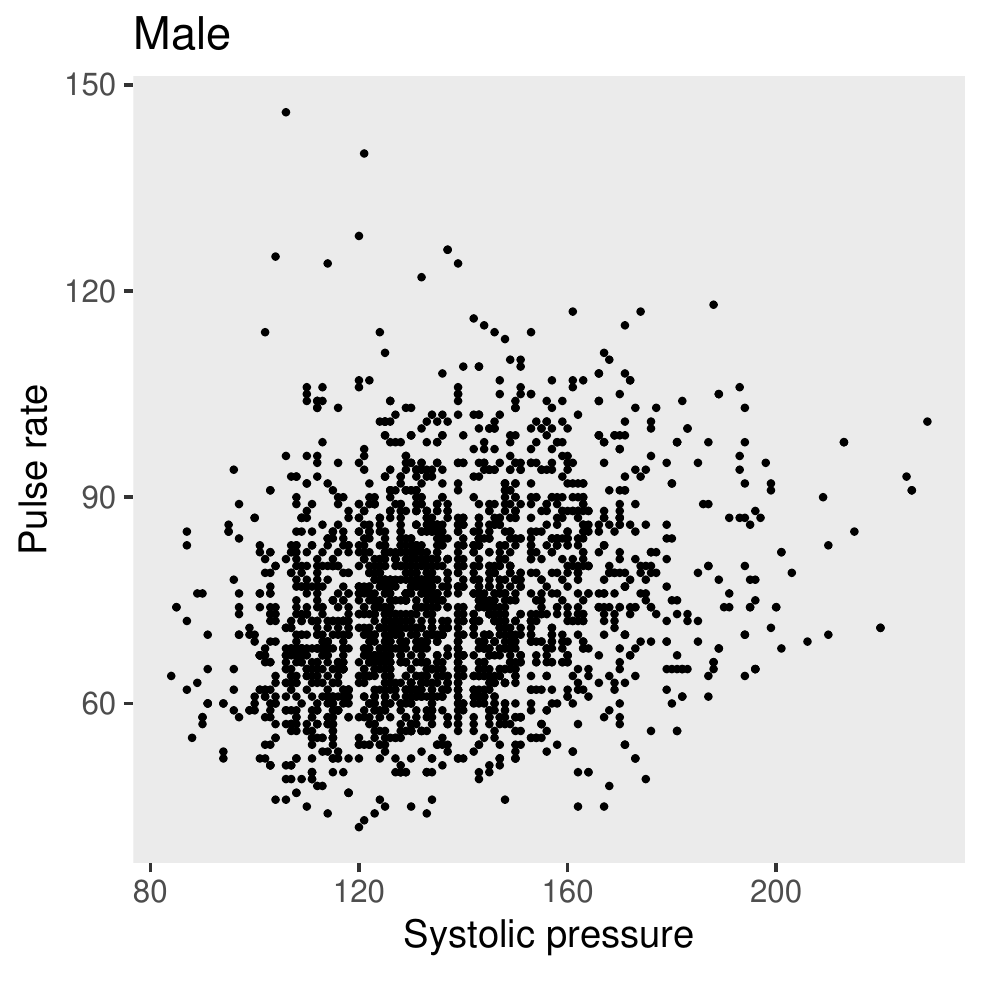}
  \includegraphics[width=7cm]{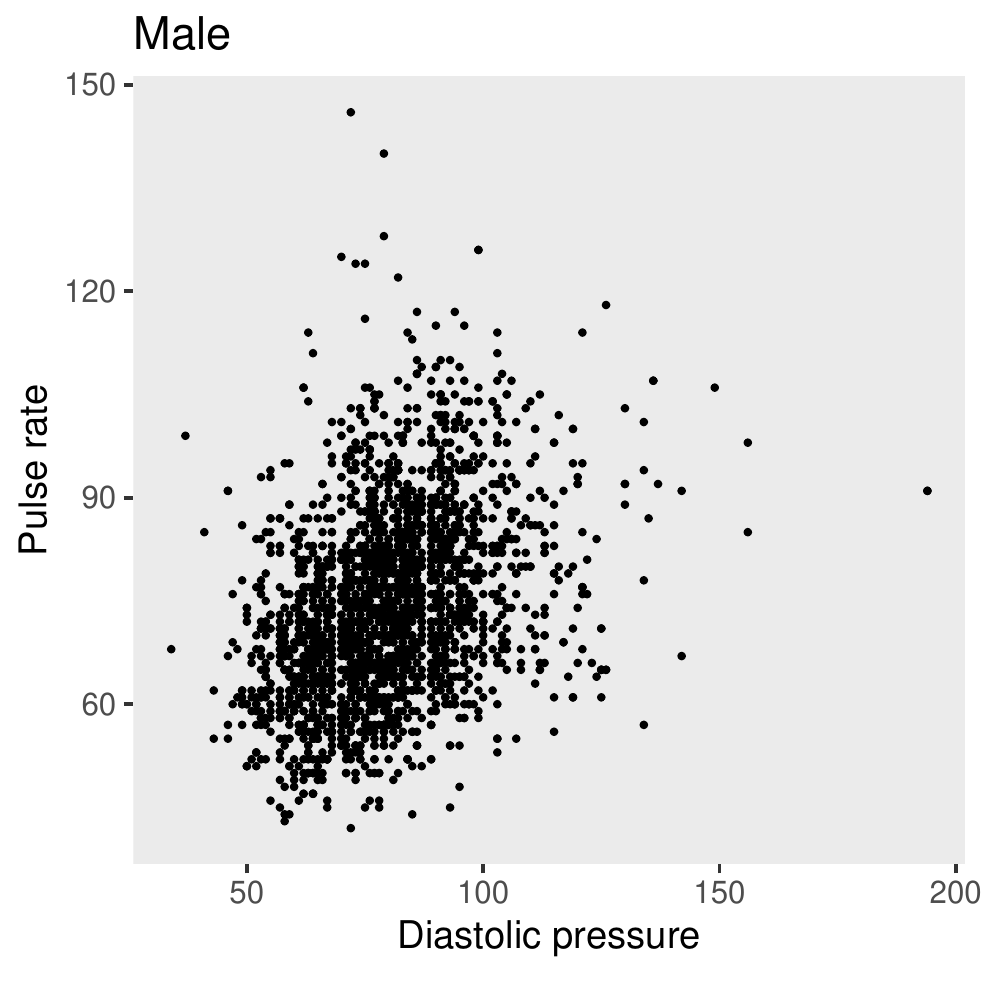}  
  \caption{The scatter plots of the pulse rates against the systolic and diastolic blood pressures for female and male patients.}
  \label{Fig:1}
\end{figure}

\begin{figure}[!htbp]
  \centering
  \includegraphics[width=7cm]{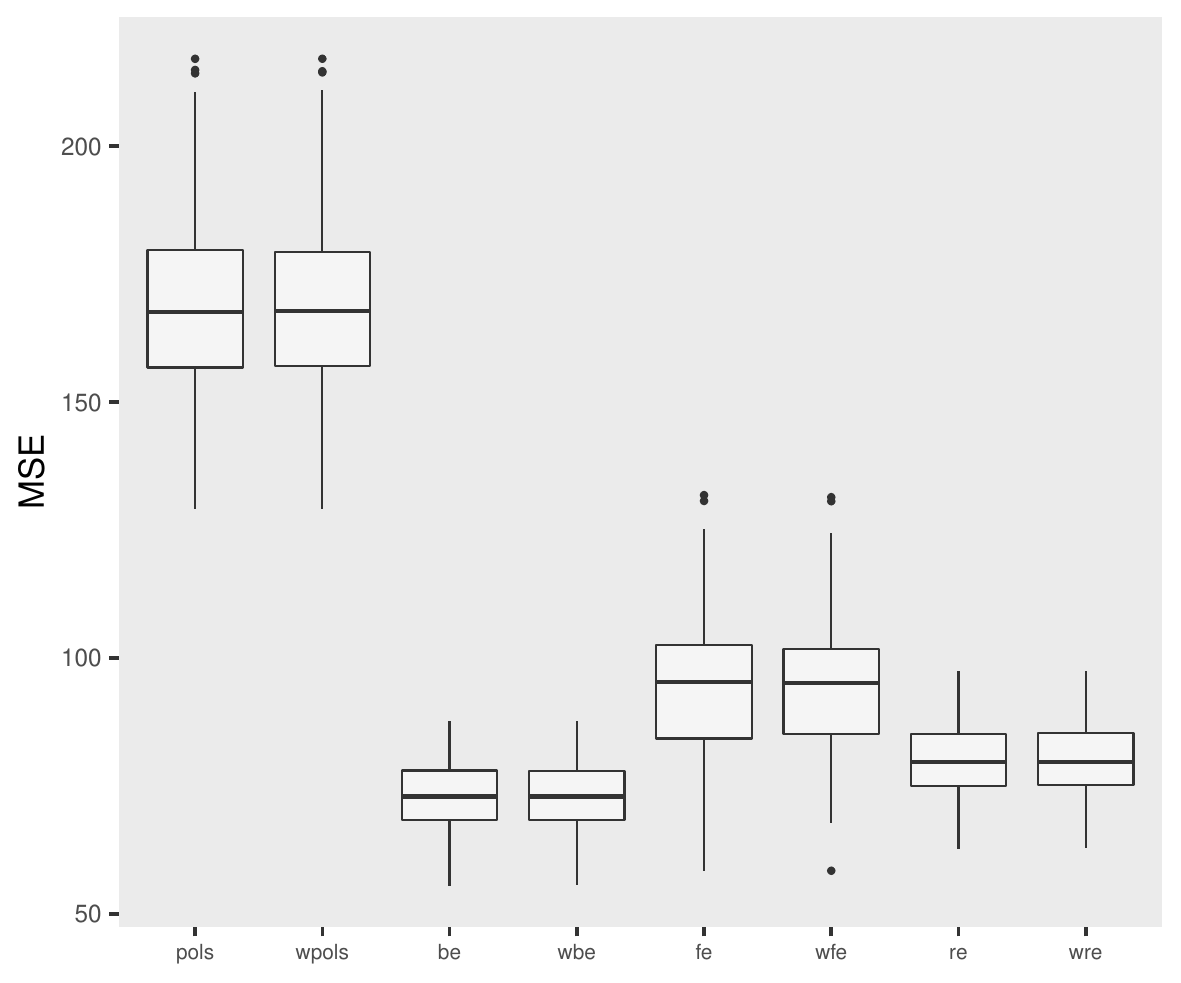}
  \includegraphics[width=7cm]{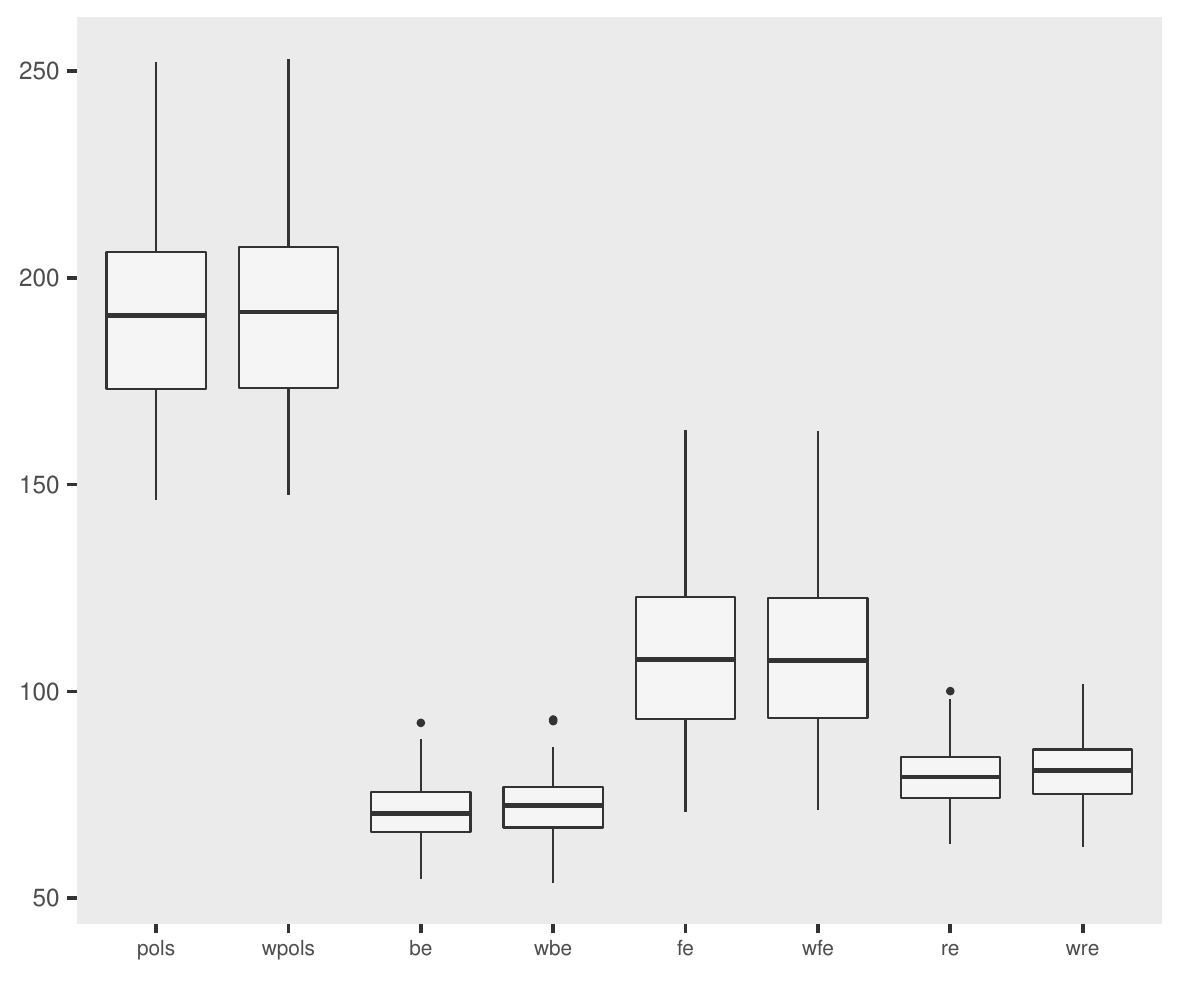}
  \caption{Calculated MSE values for the blood pressure dataset: male (first column) and female (second column). Methods: OLS based pooled regression (pols), between regression (be), fixed effects model (fe), random effects model (re), and their weighted likelihood counterparts; wpols, wbe, wfe, and wre.}
  \label{Fig:pred}
\end{figure}

\end{document}